\documentclass[11pt]{article}
\usepackage{fullpage,amsthm,enumitem}
\usepackage{amsmath,amsfonts,amssymb,graphicx} 
\usepackage[T1]{fontenc}

\usepackage{mathpazo} 
\usepackage[dvipsnames]{xcolor} 
\usepackage{hyperref}
\usepackage{tikz-cd}
\usepackage{pgfplots}
\usepackage{doi}
\usepackage{esint}
\usepackage{multicol}
\usepackage{verbatim}
\usepackage{float}
\usepackage{subcaption}
\usepackage{slashed}
\usetikzlibrary{decorations.markings}
\newtheorem{thm}{Theorem}
\newtheorem{cor}[thm]{Corollary}
\newtheorem{lem}[thm]{Lemma}
\newtheorem{prop}[thm]{Proposition}
\newtheorem{defn}[thm]{Definition}
\newtheorem{ex}[thm]{Example}
\newtheorem{rem}[thm]{Remark}

\newtheorem{notation}[thm]{Notation}

\newcommand{\vertex}[1]{
	\filldraw (#1) circle (4pt);
	\filldraw[white] (#1) circle (2pt)
	}

\tikzset{edge/.style={decorate, decoration=snake}}

\newcommand{\ncvertexh}[1]{\filldraw (#1) circle (6.25pt);
	{\color{white}\filldraw (#1) circle (5.1pt);}
	\node at (#1) {\tiny$f$};}

\newcommand{\N}{\mathbb{N}}

\newcommand{\R}{\mathbb{R}}
\newcommand{\C}{\mathbb{C}}

\renewcommand{\O}{\mathcal{O}}

\newcommand{\bfr}{\mathfrak{b}}

\newcommand{\vect}[2]{\begin{pmatrix}#1\\#2\end{pmatrix}}

\newcommand{\f}{\textnormal{fi}}
\newcommand{\Efi}{E_\textnormal{fi}}
\newcommand{\Vfi}{V_\textnormal{fi}}

\newcommand{\sgn}{\mathrm{sgn}\;}

\DeclareMathOperator{\Tr}{Tr}

\newcommand{\Ampl}{\operatorname{Ampl}}
\newcommand{\diag}{\operatorname{diag}}
\def\multinom#1#2{\ensuremath{\left(\kern-.3em\left(\genfrac{}{}{0pt}{}{#1}{#2}\right)\kern-.3em\right)}}

\begin{document}
\author{Eva-Maria Hekkelman, Teun D. H. van Nuland, Jesse Reimann}
\title{Power counting in the spectral action matrix model}
\maketitle

\begin{abstract}
We derive power counting formulas for ribbon graph amplitudes that were recently independently discovered in two contexts, namely as a generalization of the Kontsevich model, and as corresponding to a matrix model approach to the spectral action. The Feynman rules are based on divided difference functions of eigenvalues of an abstract Dirac operator. We obtain formulas for the order of divergence, depending on the spectral dimension $d$, the order of decay of the test function $f$ of the spectral action, and the graph properties. Several consequences are discussed, such as the fact that all graphs with maximal order of divergence (at a given loop order and number of external vertices) are planar. 
To derive our main results we establish two-sided bounds for divided differences, and in particular generalize Hunter's positivity theorem to a larger class of functions. 
\end{abstract}

\section{Introduction}\label{sect: intro}
Motivated by the quest for a QFT description of the spectral action~\cite{ChamseddineConnes1997} we would like to analyze the large $N$ behavior of the correlation functions
\begin{align}\label{eq:path integral}
\frac{\int_{H_N}V_{i_1j_1}\cdots V_{i_mj_m}e^{-\hbar^{-2}\Tr(f(D+\hbar V)-f(D))}d V}{\int_{H_N}e^{-\hbar^{-2}\Tr(f(D+\hbar V)-f(D))}d V},
\end{align}
where $\hbar$ is a formal parameter, $f:\R\to\R$ is a sufficiently regular function, 
$D=\diag(\lambda_k)_{k=1}^\infty$ is a self-adjoint operator diagonalized with respect to some countable orthonormal basis, and $H_N$ is the space of hermitian matrices acting on the first $N$ basis elements.\footnote{$H_N$ is equipped with the canonical measure coming from the real and imaginary parts of the components of $V$.} 

Following \cite{vNvS21,vNvS22b,vNvS23}, we write 
$$\Tr(f(D+\hbar V)-f(D))=\sum_{n=1}^\infty\sum_{i_1,\ldots,i_n=1}^N \frac{\hbar^n}{n}f'[\lambda_{i_1},\ldots,\lambda_{i_n}]V_{i_1i_2}\cdots V_{i_n i_1},$$ in terms of the divided differences of the derivative of $f$, defined inductively
by $f'[x]:=f'(x)$, $f'[x,y]:=\frac{f'(x)-f'(y)}{x-y}$, $f'[x,y,z]:=\frac{f'[x,z]-f'[y,z]}{x-y}=\frac{f'(x)-f'(z)}{(x-y)(x-z)}-\frac{f'(y)-f'(z)}{(x-y)(y-z)}$, \textit{et cetera}. By recognizing the second order term $\frac12\sum_{k,l=1}^Nf'[\lambda_k,\lambda_l]|V_{kl}|^2$ as the free theory, we can apply standard methods of Gaussian integration in order to express \eqref{eq:path integral} as a combinatorial series whose terms are ribbon graph amplitudes \cite{Eynard2016,vNvS23}. The respective ribbon graphs 
are generalized Kontsevich graphs in the terminology of~\cite{BCEG}. Although \cite{BCEG} considers polynomials $f$ (the order of which limits the valency of the vertices), the corresponding amplitudes turn out to be formally the same as in \cite{vNvS22b}: these amplitudes are certain fractions of divided differences of $f'$.

A stunning feature of these graph amplitudes, noted independently by \cite{BCEG,vNvS22b}, is the Ward--Takahashi identity, 
which \cite{vNvS22b} showed to imply one-loop renormalizability of the spectral action matrix model in the Gomis--Weinberg sense, raising the question what algebraic relations govern the graphs at arbitrary loop order, and if their renormalization flow can be understood. These questions, besides being inherently interesting \cite{Azarfar2024,BG2016,HKPV2022,Perez2022,Perez2025,Steinacker2010,tHooft1982}, prepare for replacing the integration space~$H_N$ in \eqref{eq:path integral} by a space of fields incorporating physical content such as the spectral Standard Model \cite{ChamseddineConnes1997,CC2012,CIS2020}, and beyond \cite{AMST2016,CCM2015,CCS2013,DM2017,Suijlekom2015}. Moreover, their positive answer would bring the spectral action more in line with noncommutative quantum field theory \cite{GW2005a,GW2005b,GW2014,Riv2007}.


\paragraph{Feynman rules}
We summarize the Feynman rules referred to above --
precise definitions are in Section \ref{sct:Definitions}. The Feynman diagrams are ribbon graphs for which not the edges but the faces are labeled by indices.
For each vertex bordered by faces with indices $i_1,\ldots,i_n$ in between its incident edges (possibly counting faces multiple times), we multiply by a factor $f'[\lambda_{i_1},\ldots,\lambda_{i_n}]$, e.g.,
\begin{align*}
\raisebox{-18pt}{
\begin{tikzpicture}[thick]
	\draw (-0.5,-0.5) to (0,0);
	\draw (-0.5,0.5) to (0,0);
	\draw (0.5,0.5) to (0,0);
	\draw (0.5,-0.5) to (0,0);
	\node at (-0.5,0) {$i$};
	\node at (0,0.5) {$j$};
	\node at (0.5,0) {$k$};
	\node at (0,-0.5) {$l$};
	\vertex{0,0};
\end{tikzpicture}}
\,\quad&=\quad f'[\lambda_i,\lambda_j,\lambda_k,\lambda_l],\\
\intertext{for each internal edge bordered by $i$ and $j$, we divide by a factor $f'[\lambda_i,\lambda_j]$, i.e.,}
\raisebox{-14pt}{
\begin{tikzpicture}[thick]
	\draw (-0.5,-0.5) to (0.5,0.5);
	\node at (-0.25,0.25) {$i$};
	\node at (0.25,-0.25) {$j$};
\end{tikzpicture}}
\,\quad&=\quad\frac{1}{f'[\lambda_i,\lambda_j]},\\
\intertext{and, finally, we sum over each unbroken face (face without external edges), e.g.,}
\raisebox{-18pt}{
\begin{tikzpicture}[thick]
\draw (-0.5,0) to[out=90,in=180] (0.5,0.5);
\draw (-0.5,0) to[out=-90,in=180] (0.5,-0.5);
\draw (0.5,0.5) to[out=-45,in=45] (0.5,-0.5);
\draw (-0.5,0) to (-0.8,0);
\draw (0.5,0.5) to (0.7,0.7);
\draw (0.5,-0.5) to (0.7,-0.7);
\vertex{-0.5,0};
\vertex{0.5,0.5};
\vertex{0.5,-0.5};
\end{tikzpicture}}
\,\quad&=\quad\sum_{k=1}^N
\raisebox{-18pt}{
\begin{tikzpicture}[thick]
\draw (-0.5,0) to[out=90,in=180] (0.5,0.5);
\draw (-0.5,0) to[out=-90,in=180] (0.5,-0.5);
\draw (0.5,0.5) to[out=-45,in=45] (0.5,-0.5);
\draw (-0.5,0) to (-0.8,0);
\draw (0.5,0.5) to (0.7,0.7);
\draw (0.5,-0.5) to (0.7,-0.7);
\node at (0.15,0) {$k$};
\vertex{-0.5,0};
\vertex{0.5,0.5};
\vertex{0.5,-0.5};
\end{tikzpicture}}\,\,.
\end{align*}

\begin{ex}
Two amplitudes in the spectral action matrix model are
\begin{align*}
	\raisebox{-20pt}{
\begin{tikzpicture}[thick]
\draw (0.45,-0.7) to (1,0);
\draw (1.55,-0.7) to (1,0);
\draw (1,0) arc (-90:270:0.5cm);
\node at (1,-0.55) {$i$};
\node at (0.2,0.8) {$j$};
\vertex{1,0};
\end{tikzpicture}
}
\quad=\quad\sum_{k=1}^N
	\raisebox{-20pt}{
\begin{tikzpicture}[thick]
	\draw (0.45,-.7) to (1,0);
	\draw (1.55,-.7) to (1,0);
	\draw (1,0) arc (-90:270:0.5cm);
	\node at (1,-.55) {$i$};
	\node at (0.2,0.8) {$j$};
	\vertex{1,0};
        \node at (1,.575) {$k$};
\end{tikzpicture}
}\quad=\quad\sum_{k=1}^N\frac{f'[\lambda_j,\lambda_k,\lambda_j,\lambda_i]}{f'[\lambda_j,\lambda_k]},
\end{align*}
and
\begin{align*}
	\raisebox{-20pt}{
\begin{tikzpicture}[thick]
	\draw (0.45,-.7) to (1,0);
	\draw (1.55,-.7) to (1,0);
	\draw (1,0.054) arc (-90:270:0.3cm);
	\draw (1,-0.035) arc (-90:270:.7cm);
	\node at (1,-.55) {$i$};
	\node at (0,1) {$j$};
	\vertex{1,0};
\end{tikzpicture}
}
\quad=\quad\sum_{k,l=1}^N
	\raisebox{-20pt}{
\begin{tikzpicture}[thick]
	\draw (0.45,-.7) to (1,0);
	\draw (1.55,-.7) to (1,0);
	\draw (1,0.054) arc (-90:270:0.3cm);
	\draw (1,-0.035) arc (-90:270:.7cm);
	\node at (1,-.55) {$i$};
	\node at (0,1) {$j$};
	\node at (1,1) {$k$};
	\node at (1,0.4) {$l$};
	\vertex{1,0};
\end{tikzpicture}
}\quad=\quad\sum_{k,l=1}^N\frac{f'[\lambda_j,\lambda_k,\lambda_l,\lambda_k,\lambda_j,\lambda_i]}{f'[\lambda_j,\lambda_k]f'[\lambda_k,\lambda_l]}.
\end{align*}
\end{ex}

For renormalization purposes, it is relevant \cite{CR2000,CR2001,GW2005a,GW2005b,ILV2012,KLV2014,LOR2015,Riv2007,RVW} to know the asymptotic behavior of these amplitudes as $N\to\infty$. 
The main objective of this paper is to prove the following formulas describing this asymptotic behavior. 

\paragraph{Main results}
Suppose that $\{\lambda_k\}_{k=1}^\infty\subseteq\R$ is a sequence of asymptotic behavior $\lambda_k\asymp k^{1/d}$ for $d\in\R_{>0}$. Suppose that $f$ is smooth, even, and satisfies, for $x\to\infty$, 
$f^{(n)}(x)\asymp (-1)^nx^{-p-n}$ ($n\in\N$) for some~$p>0$. Suppose moreover that the divided differences of $f$ do not vanish on $\{\lambda_k\}_{k=1}^\infty$. For a graph $G$ with $U$ unbroken faces, $E_{\f}$ edges which do not border a broken face, and $V_{\f}$ vertices which do not border a broken face, the amplitude of $G$ is bounded from above and below by a constant times $N^{\omega(G)}$, where
\begin{align}\label{eq:power counting formula}
\omega(G)= U+\frac{p}{d}(E_{\f}-V_{\f}).
\end{align}

However, when the assumption of vanishing divided differences is not satisfied, the divergences become (for certain graphs \textit{strictly}) larger. In this case, the amplitude is bounded from above by $N^{\tilde\omega(G)}$, where
\begin{align}\label{eq:power counting formula 2}
\tilde \omega(G)=\max(U^{\bfr}+\frac{p}{d}(E^{\bfr}_{\textnormal{fi}}-V^{\bfr}_{\textnormal{fi}})+\frac{p+1}{d}(E^{\bfr}_{10}-V^{\bfr}_{10})),
\end{align}
in which the maximum is taken over all subsets $\bfr$ of unbroken faces, $U^{\bfr},\Efi^b,\Vfi^b$ are as before when designating the elements of $\bfr$ as broken, and $V^{\bfr}_{10}/E^{\bfr}_{10}$ denotes the number of vertices/edges bordering only unbroken faces except for either exactly one element of $\bfr$ or exactly one external index $i_0$ satisfying~${f'(\lambda_{i_0})=0}$.


\paragraph{Consequences of \eqref{eq:power counting formula}}
Note that $\Efi-\Vfi\geq 0$ for any graph. A remarkable result is that a function~$f$ with \textit{faster} decay results in a \textit{higher} degree of divergence of the graph. This is not obvious from the definition of the Feynman rules, even for simple graphs like in Example 1, and even for concrete functions like $f(x)=x^{-p}$.
Because divided differences of such functions may well be negative, one might \textit{a priori} expect cancellations of terms damping the degree of divergence, but this does not happen.

Another corollary of our power counting formula is the following observation. Among the connected graphs with~$n\geq1$ external edges and $L$ loops, the maximal value of $U$ is $L-1$, and the maximal value of $\Efi-\Vfi$ is $L-1$. The maximal value $\omega(G)=L-1+\frac{p}{d}(L-1)$ is attained by a nonempty set of diagrams, by virtue of our lower bound. These diagrams with maximal divergence are precisely the planar diagrams that cannot be split into two connected components by removing one vertex and all external edges, and moreover have only one unbroken face -- cf.\ Remark \ref{rem:2-point 2-loop}.

\paragraph{Consequences of \eqref{eq:power counting formula 2}}
%
The situation where divided differences of $f'$ may vanish is more complicated (cf. Section \ref{sct:main2}), but \eqref{eq:power counting formula 2} leads us to conjecture that the diagrams with maximal divergence are of the same planar form as before -- cf.\ Remark \ref{rem:conjectures}. As we discuss in Remark \ref{rem:UV/IR}, the influence of modes $\lambda_i$ with $f'(\lambda_i)=0$ is reminiscent of the UV/IR behavior of scalar field theories on noncommutative spacetime, which underlines the need for a rigorous renormalization analysis of the spectral action beyond the weak field approximation.

\paragraph{Techniques} To prove the above power counting formulas, we introduce several techniques in the asymptotic analysis of divided differences, and prove upper and lower bounds for them, which appear to be novel. A pivotal concept throughout the proof is the weighted divided difference
$$f'\{x_1,\ldots,x_n\}:=(-1)^n x_1\cdots x_n f'[x_1,\ldots,x_n].$$
Its asymptotic behavior is more easily understood, because sending any of its variables to infinity yields another weighted divided difference.
Furthermore, a key result is that weighted divided differences are positive on large enough subsets of $\R^n$ for the functions $f$ we are considering. We thus generalize and give a new proof for Hunter's positivity theorem \cite{Hunter1977}, which is recovered by taking $f(x)=x^{-p}$ ($p\in 2\N$). Moreover, we show that for functions bounded above or below by $x^{-p}$, with similar bounds on their derivatives, the asymptotic behavior of the weighed divided differences away from the origin is determined by the variable with the smallest modulus, which is crucial in the proof of our main theorem.

Our main theorem is proved by combining this positivity and the mentioned upper and lower bounds with an interesting graph-theoretical argument.

\paragraph{Acknowledgements}
We are grateful to Martijn Caspers, Séverin Charbonnier, Harald Grosse, and Walter van Suijlekom for useful discussions. TvN thanks the Erwin Schr\"odinger Institute and the organizers and participants of the April 2023 conference 'Non-commutative Geometry meets Topological Recursion', where the motivation for this paper originated.
TvN was supported by NWO project ‘Noncommutative multi-linear harmonic analysis and higher order spectral shift’,
OCENW.M.22.070. EMH thanks the Max Planck Institute for Mathematics in Bonn for its financial support. Data availability statement: no data was used for this article.


















\section{Definitions}
\label{sct:Definitions}

Throughout we fix a positive real number $d\in\R_{>0}$, and a sequence $\{\lambda_k\}_{k=1}^\infty\subseteq\R$ with the property that there exist numbers $K\in\N$, $c_1,c_2\in\R_{>0}$ such that
\begin{align}\label{eq:eigenvalues asymptotics assumption}
    c_1 k^{1/d}\leq|\lambda_k|\leq c_2 k^{1/d}
\end{align}
for all $k\geq K$.
We colloquially refer to $\{\lambda_k\}_{k=1}^\infty$ as \textbf{the spectrum}, because our main example of such a sequence is the spectrum of an abstract Dirac operator in a spectral triple. Typically, such operators satisfy Weyl's law,\footnote{Weyl's law is valid for the Laplacian on a bounded domain $\Omega$ in $\R^d$, and for the Laplace--Beltrami operator on a $d$-dimensional compact Riemannian manifold with or without boundary~\cite{MinakshisundaramPleijel1949}, see also~\cite{DuistermaatGuillemin1975,Ivrii1980} for refinements.
The validity of Weyl's law in noncommutative spaces is investigated in quite some generality in~\cite{EcksteinZajac2015, McDonaldSukochev2022WeylLaw, Vassilevich2007}.} which implies our assumption \eqref{eq:eigenvalues asymptotics assumption}.

\subsection{Notation} We write $\mathbb{N}$ for the set of natural numbers including zero, and $\mathbb{N}_{\ge 1}$ when zero is excluded. We write $A\lesssim B$ if there exists a constant $c>0$ such that $A\le cB$. 

\subsection{Divided differences}

\begin{defn}\label{defn:divdif}
Let $n\in\N$. Recursively, we define the $n^{\text{th}}$ order \textbf{divided difference} of a function ${f\in C^n(\R)}$ as the unique function $f^{[n]}\in C(\R^{n+1})$ such that
\begin{align*}
	f^{[0]}(x_0)&:=f(x_0),\\
	f^{[n]}(x_0,\ldots,x_n)&:=\frac{f^{[n-1]}(x_0,\ldots,x_{n-1})-f^{[n-1]}(x_1,\ldots,x_n)}{x_0-x_n},
\end{align*}
for $x_0,\ldots,x_n\in\R$ satisfying $x_0\neq x_n$. We write $f[x_0,\dotsc,x_n]:=f^{[n]}(x_0,\dotsc,x_n)$ for all $x_j\in\R$.
\end{defn}
For any $f\in C^n(\R)$, such a function $f^{[n]}\in C(\R^{n+1})$ always exists and is unique (cf.\ \cite{Dub20}) and so the divided difference is defined for possibly coinciding $x_j$'s. For $n=1$, this follows immediately from the differentiability of $f\in C^1(\R)$, and clearly $f[x,x]=f'(x)$. It follows similarly that $f^{[k]}[x,\dotsc,x]=\frac{1}{k!}f^{(k)}(x)$ for all $k\le n$. Furthermore, it is well-known that divided differences are invariant under permutation of variables, allowing us to occasionally assume $|x_0|\le \dotsc\le|x_n|$ without loss of generality.

There are many known expressions for divided differences. 
For example, writing out the inductive formula for divided differences yields \begin{equation}\label{eqn: divdiff_written_out}
    f[x_0,\dotsc,x_n]=\sum_{j=0}^n\frac{f(x_j)}{\prod_{k\neq j}(x_j-x_k)}\quad\text{if } x_j\neq x_k\text{ for all }j\neq k.
\end{equation}
We will furthermore be making use of a well-known integral expression for $f\in C^n(\R)$ (which follows by induction from Definition \ref{defn:divdif}, cf.~\cite[Proof of Lemma 5.1]{PSS2013} or \cite[Formula~(7.12)]{DevoreLorentz1993}), namely
\begin{equation}\label{eqn: divdiff_integral_form}
    f[x_0,\dotsc,x_n]=\int_0^1 dt_{n}\int_0^{t_{n}}dt_{n-1}\cdots \int_0^{t_2}dt_1 f^{(n)}(x_n+(x_{n-1}-x_n)t_{n}+\cdots+(x_0-x_1)t_1).
\end{equation}
By applying the mean value theorem for integrals, as a consequence we obtain that
\[
f[x_0,\ldots,x_n] = \frac{1}{n!}f^{(n)}(\xi), 
\]
for some $\xi\in [\min_{0\leq i\leq n} x_i, \max_{0\leq i\leq n} x_i]$.
\subsection{Ribbon graphs and their amplitudes}
\subsubsection{Feynman ribbon graphs}
Multigraphs are generalized graphs which are allowed to have self-loops and multiple edges between the same two vertices. 
\begin{defn}A \textbf{multigraph} is a tuple $G=(G^0,G^1)$ consisting of a set of vertices $G^0$ and a multiset $G^1$ of unordered pairs of elements of $G^0$ called edges.
\end{defn}
By definition of a multiset, an unordered pair $\{v_1,v_2\}\subseteq G^0$ may have multiple instances in~$G^1$, and those distinct instances $e\in G^1$ we call edges. For any edge $e$ that is an instance of an unordered pair $\{v_1,v_2\}$, we say that $e$ is \textit{incident} to $v_1$ and $v_2$.

\begin{figure}[H]
\centering
\raisebox{-20pt}
{\begin{tikzpicture}
	\draw[out=-90,in=-135] (0,0) to (1,0);
    \draw[out=90,in=135] (0,0) to (1,0);
	\draw (1,0) to (2,0);
    \draw (2,0) to[out=0,in=-100] (3,0.5);
    \draw (2,0) to[out=80,in=170] (3,0.5);
    \draw (2,0) to (3,-0.5);
    \draw[in=60, out = -60] (3,0.5) to (3,-0.5);
	\draw[in=90, out=90] (1,0) to (3,0.5);
    \draw[in = -90, out=-90] (1,0) to (3,-0.5);
    \filldraw (1,0) circle (2pt);
    \filldraw (2,0) circle (2pt);
    \filldraw (3,0.5) circle (2pt);
    \filldraw (3,-0.5) circle (2pt);
\end{tikzpicture}}
\caption{A multigraph.}
\end{figure}
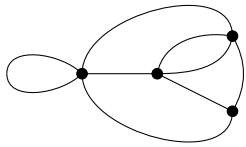

The objects called Feynman diagrams in QFT are essentially finite multigraphs, except that they may also have edges that attach on one side to a number in $\{1,\ldots,n\}$.
\begin{defn} A \textbf{Feynman diagram} is a triple $G=(G^0,n,G^1)$ consisting of a finite set of vertices $G^0$, a number $n\in\N$, and a finite multiset $G^1$ of unordered pairs of elements of $G^0\sqcup\{1,\ldots,n\}$ called edges, in such a way that for every $i\in\{1,\ldots,n\}$, $\{v,i\}$ occurs precisely once in $G^1$ for precisely one $v\in G^0$, and~$\{i,i\}$ does not occur in $G^1$.
An edge incident to an $i\in\{1,\ldots,n\}$ is called \textit{external}, and all other edges in~$G^1$ are called \textit{internal}. The set of internal edges is denoted by $G^1_{\textnormal{int}}$.
\end{defn}
 
\begin{figure}[H]
\begin{center}
\raisebox{-20pt}
{\begin{tikzpicture}
	\draw[out=-90,in=-135] (0,0) to (1,0);
    \draw[out=90,in=135] (0,0) to (1,0);
	\draw (1,0) to (2,0);
    \draw (2,0) to[out=0,in=-100] (3,0.5);
    \draw (2,0) to[out=80,in=170] (3,0.5);
    \draw (2,0) to (3,-0.5);
    \draw[in=60, out = -60] (3,0.5) to (3,-0.5);
    \draw (3,0.5) to (4,1);
    \draw (3,-0.5) to (4, -1);
	\draw[in=90, out=90] (1,0) to (3,0.5);
    \draw[in = -90, out=-90] (1,0) to (3,-0.5);
    \draw (2,0) to (2,-0.3);
    \node at (2,-0.5) {3};
    \filldraw (1,0) circle (2pt);
    \filldraw (2,0) circle (2pt);
    \filldraw (3,0.5) circle (2pt);
    \filldraw (3,-0.5) circle (2pt);
	\node at (4.2,1.1) {1};
    \node at (4.2,-1.1) {2};
\end{tikzpicture}}
\end{center}
\caption{A Feynman diagram $(G^0,3,G^1)$}
\end{figure}
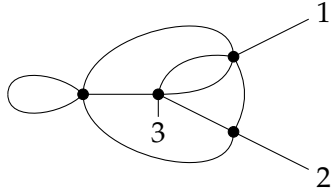

\begin{defn}
A \textbf{Feynman ribbon graph} is a Feynman diagram $G$ whose every vertex is decorated with a cyclic order on its incidences (i.e. on the multiset of edges that are incident to it, with multiplicities for self-loops). 
\end{defn}
\begin{figure}[H]
\begin{center}
\raisebox{-20pt}
{\begin{tikzpicture}[thick]
	\draw[out=-90,in=-135] (0,0) to (1,0);
    \draw[out=90,in=135] (0,0) to (1,0);
	\draw (1,0) to (2,0);
    \draw (2,0) to (3,0.5);
    \draw (2,0) to[out=0,in=-90] (2.5,0.2);
    \draw (2.5,0.3) to[out=90,in=170] (3,0.5);
    \draw (2,0) to (3,-0.5);
    \draw[in=60, out = -60] (3,0.5) to (3,-0.5);
    \draw (3,0.5) to (4,1);
    \draw (3,-0.5) to (4, -1);
	\draw[in=90, out=90] (1,0) to (3,0.5);
    \draw[in = -90, out=-90] (1,0) to (3,-0.5);
    \draw (2,0) to (2,-0.3);
    \node at (2,-0.5) {3};
    \vertex{1,0};
    \vertex{2,0};
    \vertex{3,0.5};
    \vertex{3,-0.5};
	\node at (4.2,1.1) {1};
    \node at (4.2,-1.1) {2};
\end{tikzpicture}}
\end{center}
\caption{A Feynman ribbon graph}
\end{figure}
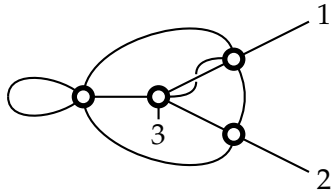

We only consider Feynman ribbon graphs up to isomorphism\footnote{This is implicit in \cite{vNvS22b,vNvS23}.}, defined as follows. 
\begin{defn}
For two Feynman ribbon graphs $G=(G^0,n,G^1)$ and $H=(H^0,n,H^1)$, a pair of functions
$$(\varphi_V,\varphi_E),\quad \varphi_V:G^0\to H^0,\quad\varphi_E:G^1\to H^1,$$
is called an isomorphism if:
\begin{itemize}
\item $\varphi_V$ and $\varphi_E$ are bijective;
\item if $(e_1,\ldots,e_k)\in (G^1)^{\times k}$ is a list of all edges incident to a vertex $v\in G^0$ written in the cyclic order associated to $v$, then $(\varphi_E(e_1),\ldots,\varphi_E(e_k))\in (H^1)^{\times k}$ is a list of all edges incident to $\varphi_V(v)\in H^0$ written in the cyclic order associated to $\varphi_V(v)$;
\item for all $i\in\{1,\ldots,n\}$, if $e\in G^1$ is incident to $i$ then $\varphi_E(e)\in H^1$ is incident to $i$.
\end{itemize}
\end{defn}
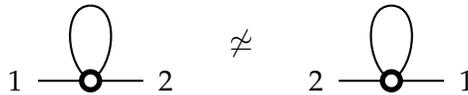
\begin{figure}[H]
\centering
\begin{tikzpicture}[thick]
\draw[out=30,in=0] (0,0) to (0,1);
\draw[out=150,in=180] (0,0) to (0,1);
\draw (-0.7,0) to (0,0);
\draw (0.7,0) to (0,0);
\vertex{0,0};
\node at (-1,0) {1};
\node at (1,0) {2};
\node at (2,0.5) {$\not\simeq$};
\draw[out=30,in=0] (4,0) to (4,1);
\draw[out=150,in=180] (4,0) to (4,1);
\draw (3.3,0) to (4,0);
\draw (4.7,0) to (4,0);
\vertex{4,0};
\node at (3,0) {2};
\node at (5,0) {1};
\end{tikzpicture}
\caption{Two nonisomorphic Feynman ribbon graphs}
\label{fig:enter-label}
\end{figure}

\subsubsection{Surface embeddings, genus, and faces}
Feynman ribbon graphs are defined abstractly, but can also be embedded without crossings on a surface, by which we mean an oriented compact two-dimensional manifold. By embedded without crossings we mean we can injectively map edges to open curves on the surface and vertices to copies of closed disks on the surface, in a way that i) preserves incidence (the union of the image of an edge $e$ and the image of a vertex $v$ is connected if and only if $e$ is incident to $v$), ii) preserves cyclic ordering (how the boundary of each disk hits the boundaries of the curves in clockwise direction), and iii) such that all disks and curves are disjoint, so in particular, none of the curves (edges) cross. This is possible because, out of an embedding with crossings, we can make an embedding without crossings by adding handles to the surface, cf.\ Figure \ref{fig:genus1b}. The \textit{genus} of a connected Feynman ribbon graph $G$ is the smallest number $g$ such that there exists a connected surface $X$ of genus $g$ on which $G$ can be embedded without crossings. If $\epsilon(G)\subseteq X$ denotes the corresponding image of~$G$ in $X$, the connected components of $X\setminus\epsilon(G)$ are called the \textit{faces} of $G$, and $F$ denotes the number of faces.
\begin{figure}
\centering
\begin{subfigure}[t]{0.23\textwidth}
\scalebox{.7}{
\begin{tikzpicture}[very thick]
\draw (-2.2,0) to (-1.5,0);
\draw (-1.5,0) arc (180:-180:1.5cm);
\draw (1.5,0) to (2.2,0);
\draw[out=45,in=135] (-1.5,0) to (0,0);
\draw[out=-45,in=-135] (0,0) to (1.5,0);
\draw[out=-45,in=-139] (-1.5,0) to (-0.1,-0.1);
\draw[out=41,in=135] (0.1,0.1) to (1.5,0);
\vertex{-1.5,0};
\vertex{1.5,0};
\draw[color=white] (-2.5,0) arc (180:-180:2.5cm);
\end{tikzpicture}
}
\caption{A Feynman ribbon graph $G$.}
\label{fig:genus1a}
\end{subfigure}
~~
\begin{subfigure}[t]{0.23\textwidth}
\scalebox{.7}{
\begin{tikzpicture}
\draw[very thick] (-2.2,0) to (-1.5,0);
\draw[very thick] (-1.5,0) arc (180:-180:1.5cm);
\draw[very thick] (1.5,0) to (2.2,0);
\draw (-0.5,0.1) to[out=75,in=180] (-0.4,0.2);
\draw (-0.4,0.2) to[out=0,in=90] (-0.2,0);
\draw (-0.2,0) to[out=-90,in=30] (-0.5,-0.5);
\draw (-0.5,0.7) to[out=-60,in=180] (-0.3,0.6);
\draw (-0.3,0.6) to[out=0,in=100] (0.25,0.25);
\draw (0.25,0.25) to[out=-80,in=110] (0.3,-0.25);
\draw (0.3,-0.25) to[out=-70,in=165] (0.55,-0.5);
\draw[very thick] (-1.5,0) to[out=45,in=180] (-0.4,0.4);
\draw[very thick] (-0.4,0.4) to[out=0,in=90] (0.05,0);
\draw[very thick] (0.05,0) to[out=-90,in=180] (0.4,-0.8);
\draw[very thick] (0.4,-0.8) to[out=0,in=-135] (1.5,0);
\draw[very thick] (-1.5,0) to[out=-45,in=-135] (-0.2,0);
\draw[very thick] (0.25,0.25) to[out=45,in=135] (1.5,0);
\vertex{-1.5,0};
\vertex{1.5,0};
\draw (-2.5,0) arc (180:-180:2.5cm);
\end{tikzpicture}
}
\caption{The graph $G$ embedded without crossings on a sphere with a handle.} 
\label{fig:genus1b}
\end{subfigure}
~~
\begin{subfigure}[t]{0.23\textwidth}
\scalebox{.7}{
\begin{tikzpicture}
\draw[very thick] (-1.5,0) arc (180:-180:1.5cm);
\draw (-0.5,0.1) to[out=75,in=180] (-0.4,0.2);
\draw (-0.4,0.2) to[out=0,in=90] (-0.2,0);
\draw (-0.2,0) to[out=-90,in=30] (-0.5,-0.5);
\draw (-0.5,0.7) to[out=-60,in=180] (-0.3,0.6);
\draw (-0.3,0.6) to[out=0,in=100] (0.25,0.25);
\draw (0.25,0.25) to[out=-80,in=110] (0.3,-0.25);
\draw (0.3,-0.25) to[out=-70,in=165] (0.55,-0.5);
\draw[very thick] (-1.5,0) to[out=45,in=180] (-0.4,0.4);
\draw[very thick] (-0.4,0.4) to[out=0,in=90] (0.05,0);
\draw[very thick] (0.05,0) to[out=-90,in=180] (0.4,-0.8);
\draw[very thick] (0.4,-0.8) to[out=0,in=-135] (1.5,0);
\draw[very thick] (-1.5,0) to[out=-45,in=-135] (-0.2,0);
\draw[very thick] (0.25,0.25) to[out=45,in=135] (1.5,0);
\draw (-2.5,0) arc (180:-180:2.5cm);
\draw[very thick] (-1.5,0) to[out=180,in=-100] (-1.9,0.8);
\draw[very thick,dotted] (-1.9,0.8) to[out=80,in=180] (0,2.2);
\draw[very thick] (1.5,0) to[out=0,in=-80] (1.9,0.8);
\draw[very thick,dotted] (1.9,0.8) to[out=100,in=0] (0,2.2);
\vertex{-1.5,0};
\vertex{1.5,0};
\node at (0,2.2) {\Large$\times$};
\end{tikzpicture}
}
\caption{The graph $G'$.}
\label{fig:genus1c}
\end{subfigure}
~~
\begin{subfigure}[t]{0.23\textwidth}
\scalebox{.7}{
\begin{tikzpicture}
\draw[very thick] (-1.5,0) arc (180:-180:1.5cm);
%
\draw (-0.5,0.1) to[out=75,in=180] (-0.4,0.2);
\draw (-0.4,0.2) to[out=0,in=90] (-0.2,0);
\draw (-0.2,0) to[out=-90,in=30] (-0.5,-0.5);
\draw (-0.5,0.7) to[out=-60,in=180] (-0.3,0.6);
\draw (-0.3,0.6) to[out=0,in=100] (0.25,0.25);
\draw (0.25,0.25) to[out=-80,in=110] (0.3,-0.25);
\draw (0.3,-0.25) to[out=-70,in=165] (0.55,-0.5);
\draw[very thick] (-1.5,0) to[out=45,in=180] (-0.4,0.4);
\draw[very thick] (-0.4,0.4) to[out=0,in=90] (0.05,0);
\draw[very thick] (0.05,0) to[out=-90,in=180] (0.4,-0.8);
\draw[very thick] (0.4,-0.8) to[out=0,in=-135] (1.5,0);
\draw[very thick] (-1.5,0) to[out=-45,in=-135] (-0.2,0);
\draw[very thick] (0.25,0.25) to[out=45,in=135] (1.5,0);
\draw[very thick] (-1.5,0) to[out=180,in=-100] (-1.9,0.8);
\draw[very thick,dotted] (-1.9,0.8) to[out=80,in=180] (0,2.2);
\draw[very thick] (1.5,0) to[out=0,in=-80] (1.9,0.8);
\draw[very thick,dotted] (1.9,0.8) to[out=100,in=0] (0,2.2);
\node at (0,2.2) {\Large$\times$};
\vertex{-1.5,0};
\vertex{1.5,0};
\draw (-2.5,0) arc (180:-180:2.5cm);
\node at (-1.65,0.6) {$1$};
\node at (-1,1.6) {$1$};
\node at (1,1.6) {$1$};
\node at (1.65,0.6) {$1$};
\node at (-1.9,-0.4) {$2$};
\node at (-1,-1.7) {$2$};
\node at (1,-1.7) {$2$};
\node at (1.9,-0.4) {$2$};
\node at (0,1) {$3$};
\node at (-0.9,0) {$3$};
\node at (0.9,0) {$3$};
\node at (-0.2,-1) {$3$};
\end{tikzpicture}
}
\caption{A choice of index numbers labeling the faces of $G'$. The vertices and edges of $G$ inherit these index numbers.}
\label{fig:genus1d}
\end{subfigure}
\caption{Surface embedding and attaching index numbers.}
\end{figure}
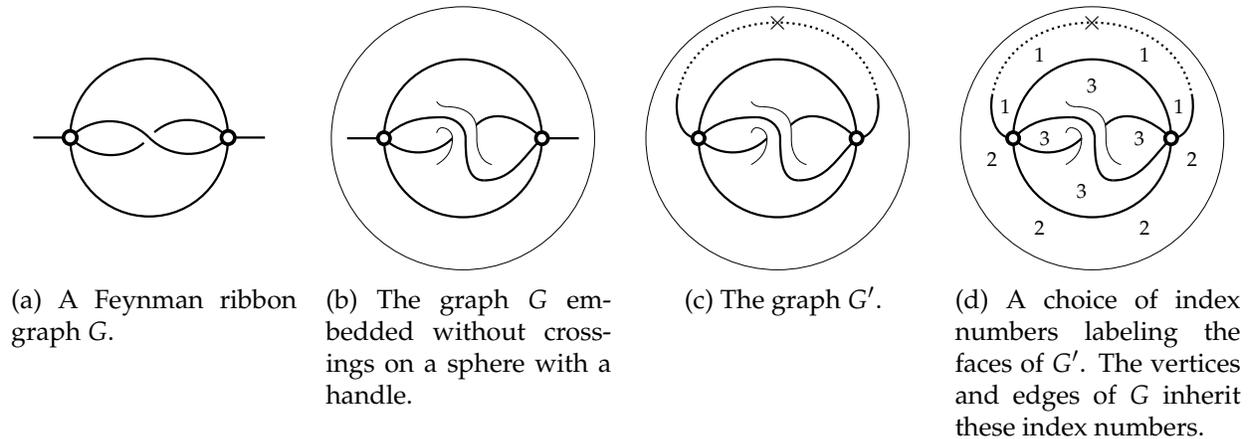
The set of faces of a nonconnected Feynman ribbon graph $G$ is the disjoint union of the sets of faces of the connected components of $G$.
A face is called \textit{broken} if at least one external edge lies in its boundary, and \textit{unbroken} otherwise. For instance, the graph of Figure \ref{fig:genus1a} has one broken face (outside) and one unbroken face (inside), as seen in Figure \ref{fig:genus1b}.

\subsubsection{Index numbers}\label{sect: Feynman_rules}
Given a Feynman ribbon graph $G$, we construct a Feynman ribbon graph $G'$ by adding to $G$ one vertex for every broken face of $G$, and attaching all external edges of that face, instead of to an element of $\{1,\ldots,n\}$, to the same vertex -- see Figure \ref{fig:genus1c}. There is a unique cyclic ordering to put on (the sets of incidences of) the new vertices in $G'$ such that the genus of $G'$ is equal to the genus of $G$, for each connected component. Thus, the Feynman ribbon graph $G'$ is determined uniquely by $G$. For bookkeeping purposes we number the faces of $G'$ by the \textit{index numbers} $1,\ldots,n+U$, where the faces $1,\ldots,n$ of~$G'$ correspond to the numbers of the external edges of $G$, and the faces $n+1,\ldots,n+U$ of $G'$ are the unbroken faces in $G$ -- see Figure \ref{fig:genus1d}.

This construction determines,
\begin{itemize}
\item for every edge $e\in G^1$, the two index numbers $\beta_1^e,\beta_2^e\in\{1,\ldots,n+U\}$ of the faces of $G'$ which border $e$, and,
\item for every vertex $v\in G^0$, the $\deg(v)$ index numbers $\alpha_1^v,\ldots,\alpha_{\deg(v)}^v\in\{1,\ldots,n+U\}$ of the faces of $G'$ which border $v$ in between incident edges, in cyclic order and including multiplicity.
\end{itemize}
In the following, we assume that every Feynman ribbon graph $G$ carries a distinguished collection of index numbers $\{\alpha_j^v\}_{v\in G^0,j\in\{1,\ldots,\deg(v)\}}\sqcup\{\beta_j^e\}_{e\in G^1,j\in\{1,2\}}$.

\subsubsection{The amplitude}
In accordance with the Feynman rules introduced in Section~\ref{sect: intro}, we define the amplitude of a Feynman ribbon graph $G$ as a function of the \textit{cutoff scale} $N\in\N$ and the \textit{external indices}~${i_1,\ldots,i_n\in\N_{\geq1}}$ as follows:
\begin{equation}\label{eqn: ampl_normal_divdiff}
    \Ampl_{N,i_1,\ldots,i_n}(G):=\sum_{i_{n+1},\ldots,i_{n+U}=1}^N\frac{\prod_{v\in G^0} f'[i(\alpha^v_1),\ldots,i(\alpha^v_{\deg(v)})]}{\prod_{e\in G^1_{\textnormal{int}}} f'[i(\beta^e_1),i(\beta^e_2)]},
\end{equation}
where we write $i(\alpha):=i_\alpha$ and $f'[j_1,\ldots,j_{n+1}]:=(f')^{[n]}(\lambda_{j_1},\ldots,\lambda_{j_{n+1}})$. We assume that $f$ is smooth and real-valued (denoted $f\in C^\infty(\R)_\R$) and satisfies $f'[i,j]\neq 0$ for all $i,j\in\mathbb{N}_{\ge1}$. It should be understood that smoothness can be replaced by ${f\in C^{m+1}(\R)}$ for $m=\max_{v\in G^0}\deg(v)$. As divided differences are invariant under permutation of their variables, we do not always list input variables in their cyclic order arising from the graph structure. For brevity we do not divide by the symmetry factor (which is 1 modulo vacuum graphs).

\begin{ex}\label{Ex:BasicAmplitude}
Upon a choice of index numbers, the above defines, for instance,
\begin{align*}
\Ampl_{N,i_1,i_2}\Bigg(
\raisebox{-25pt}
{\begin{tikzpicture}[thick]
\draw (-1.4,0) to (-1,0);
\draw (1.4,0) to (1,0);
\draw (-1,0) arc(180:-180:1);
\draw[out=45,in=135] (-1,0) to (0,0);
\draw[out=-45,in=-135] (0,0) to (1,0);
\draw[out=-45,in=-139] (-1,0) to (-0.1,-0.1);
\draw[out=41,in=135] (0.1,0.1) to (1,0);
\vertex{-1,0};
\vertex{1,0};
\node at (-1.6,0) {1};
\node at (1.6,0) {2};
\end{tikzpicture}}
\Bigg)&=\sum_{k=1}^N
\raisebox{-25pt}{
\begin{tikzpicture}[thick]
\draw (-1.8,0) to (-1,0);
\draw (1.8,0) to (1,0);
\draw (-1,0) arc(180:-180:1);
\draw[out=45,in=135] (-1,0) to (0,0);
\draw[out=-45,in=-135] (0,0) to (1,0);
\draw[out=-45,in=-139] (-1,0) to (-0.1,-0.1);
\draw[out=41,in=135] (0.1,0.1) to (1,0);
\vertex{-1,0};
\vertex{1,0};
\node at (-1.4,0.4) {$i_1$};
\node at (-1.4,-0.4) {$i_2$};
\node at (0,0.5) {$k$};
\end{tikzpicture}}
=\sum_{k=1}^N\frac{f'[i_1,k,k,k,i_2]^2}{f'[i_1,k]f'[k,k]^2f'[k,i_2]}.
\end{align*}
\end{ex}
\begin{notation}
When an index is not specified, we sum over it from $1$ to $N$, e.g.,
\begin{align*}
\raisebox{-25pt}
{\begin{tikzpicture}[thick]
\draw (-1.8,0) to (-1,0);
\draw (1.8,0) to (1,0);
\draw (-1,0) arc(180:-180:1);
\draw[out=45,in=135] (-1,0) to (0,0);
\draw[out=-45,in=-135] (0,0) to (1,0);
\draw[out=-45,in=-139] (-1,0) to (-0.1,-0.1);
\draw[out=41,in=135] (0.1,0.1) to (1,0);
\vertex{-1,0};
\vertex{1,0};
\node at (-1.4,0.4) {$i_1$};
\node at (-1.4,-0.4) {$i_2$};
\end{tikzpicture}}
&=\sum_{k=1}^N
\raisebox{-25pt}{
\begin{tikzpicture}[thick]
\draw (-1.8,0) to (-1,0);
\draw (1.8,0) to (1,0);
\draw (-1,0) arc(180:-180:1);
\draw[out=45,in=135] (-1,0) to (0,0);
\draw[out=-45,in=-135] (0,0) to (1,0);
\draw[out=-45,in=-139] (-1,0) to (-0.1,-0.1);
\draw[out=41,in=135] (0.1,0.1) to (1,0);
\vertex{-1,0};
\vertex{1,0};
\node at (-1.4,0.4) {$i_1$};
\node at (-1.4,-0.4) {$i_2$};
\node at (0,0.5) {$k$};
\end{tikzpicture}}.
\end{align*}
\end{notation}

\subsubsection{Comparison to the literature}
\begin{rem}\label{rmk:matrix models}
Ribbon graphs appearing in other matrix models (e.g., \cite{BH2023,GW2005a,tHooft1974,Perez2022,Steinacker2010}) may generally have a different index attached to each side of each edge at both places of attachment to a vertex, like so:
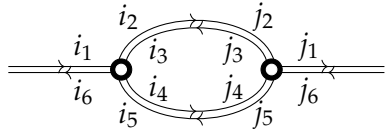
\begin{align*}
\begin{tikzpicture}[decoration={
    markings,
    mark=at position 0.5 with {\arrow{>}}}
    ] 
\draw[postaction={decorate}] (2,1.96) to (0.5,1.96);
\draw[postaction={decorate}] (0.5,2.04) to (2,2.04);
\draw[postaction={decorate}] (3.95,2) to[out=90,in=90] (2.05,2);
\draw[postaction={decorate}] (1.95,2.05) to[out=80,in=100] (4.05,2.05);
\draw[postaction={decorate}] (2.05,2) to[out=-90,in=-90] (3.95,2);
\draw[postaction={decorate}] (4.05,1.95) to[out=260,in=280] (1.95,1.95);
\draw[postaction={decorate}] (4,2.04) to (5.5,2.04);
\draw[postaction={decorate}] (5.5,1.96) to (4,1.96);
\node at (1.5,2.3) {$i_1$};
\node at (2.1,2.7) {$i_2$};
\node at (2.5,2.25) {$i_3$};
\node at (2.5,1.75) {$i_4$};
\node at (2.1,1.4) {$i_5$};
\node at (1.5,1.7) {$i_6$};
\node at (4.5,2.3) {$j_1$};
\node at (3.9,2.7) {$j_2$};
\node at (3.5,2.25) {$j_3$};
\node at (3.5,1.75) {$j_4$};
\node at (3.9,1.4) {$j_5$};
\node at (4.5,1.7) {$j_6$};
\vertex{2,2};
\vertex{4,2};
\end{tikzpicture}.
\end{align*}
In the spectral action matrix model \cite{vNvS23}, a graph like the one above only contributes a nonzero amplitude if $i_1=i_2$, $i_3=i_4$, $i_5=i_6$ and $j_1=j_2$, $j_3=j_4$, $j_5=j_6$, and moreover only if $i_2=j_2$, $i_3=j_3$, $i_4=j_4$, and $i_5=j_5$.
If these equalities hold, we denote the graph as
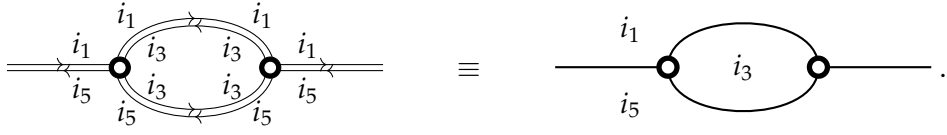
\begin{align*}
\raisebox{-22pt}{
\begin{tikzpicture}[decoration={
    markings,
    mark=at position 0.5 with {\arrow{>}}}
    ] 
\draw[postaction={decorate}] (2,1.96) to (0.5,1.96);
\draw[postaction={decorate}] (0.5,2.04) to (2,2.04);
\draw[postaction={decorate}] (3.95,2) to[out=90,in=90] (2.05,2);
\draw[postaction={decorate}] (1.95,2.05) to[out=80,in=100] (4.05,2.05);
\draw[postaction={decorate}] (2.05,2) to[out=-90,in=-90] (3.95,2);
\draw[postaction={decorate}] (4.05,1.95) to[out=260,in=280] (1.95,1.95);
\draw[postaction={decorate}] (4,2.04) to (5.5,2.04);
\draw[postaction={decorate}] (5.5,1.96) to (4,1.96);
\node at (1.5,2.3) {$i_1$};
\node at (2.1,2.7) {$i_1$};
\node at (2.5,2.25) {$i_3$};
\node at (2.5,1.75) {$i_3$};
\node at (2.1,1.4) {$i_5$};
\node at (1.5,1.7) {$i_5$};
\node at (4.5,2.3) {$i_1$};
\node at (3.9,2.7) {$i_1$};
\node at (3.5,2.25) {$i_3$};
\node at (3.5,1.75) {$i_3$};
\node at (3.9,1.4) {$i_5$};
\node at (4.5,1.7) {$i_5$};
\vertex{2,2};
\vertex{4,2};
\end{tikzpicture}
}
\qquad\equiv\qquad
\raisebox{-19pt}{
\begin{tikzpicture}[thick]
\draw (0.5,2) to (2,2);
\draw (2,2) to[out=90,in=90] (4,2);
\draw (2,2) to[out=-90,in=-90] (4,2);
\draw (4,2) to (5.5,2);
\node at (1.5,2.5) {$i_1$};
\node at (3,2) {$i_3$};
\node at (1.5,1.5) {$i_5$};
\vertex{2,2};
\vertex{4,2};
\end{tikzpicture}
}
.
\end{align*}

\end{rem}
\begin{rem}
Our graphs and their amplitudes directly correspond to the ones in \cite{vNvS22b,vNvS23}, where however edges are denoted with wavy lines to emphasize that the graphs originate from the \textbf{bosonic} spectral action. We use straight lines for visual clarity. The other difference is that we focus on the \textbf{matrix elements} of the graphs, that is, in terms of the notation of \cite[Section 4]{vNvS23},
\begin{align*}
\sum_{i,j}V_{ij}W_{ji}\Big(
\raisebox{-15pt}{
\begin{tikzpicture}[thick]
\draw (-0.7,2) to (0,2);
\draw (0,2) to[out=90,in=90] (1.5,2);
\draw (0,2) to[out=-90,in=-90] (1.5,2);
\draw (1.5,2) to (2.2,2);
\node at (-0.3,2.3) {$i$};
\node at (-0.3,1.7) {$j$};
\vertex{0,2};
\vertex{1.5,2};
\end{tikzpicture}}
\Big)\quad\equiv\quad
\raisebox{-15pt}{
\begin{tikzpicture}
\draw[edge] (-0.8,2) to (0,2);
\draw[edge] (-0.09,1.95) to[out=90,in=90] (1.49,1.97);
\draw[edge] (0,2) to[out=-90,in=-90] (1.5,2);
\draw[edge] (1.5,2) to (2.3,2);
\node at (-1.1,2) {$V$};
\node at (2.7,2) {$W$~.};
\ncvertexh{0,2}
\ncvertexh{1.5,2}
\end{tikzpicture}}
\end{align*}
\end{rem}

\subsubsection{Spectral translational invariance}
\begin{rem}\label{rem:translational invariance}
Let us write $\Ampl_{N,i_1,\ldots,i_n}^{f,\lambda}(G)$ to indicate the dependence on $f$ and $\lambda=\{\lambda_k\}_{k=1}^\infty$. For~${\epsilon\in\R}$, defining $\lambda^\epsilon_k:=\lambda_k+\epsilon$ and $f^\epsilon(x):=f(x-\epsilon)$, we note that
$$\Ampl_{N,i_1,\ldots,i_n}^{f^\epsilon,\lambda^\epsilon}(G)=\Ampl_{N,i_1,\ldots,i_n}^{f,\lambda}(G).$$
In particular, this allows us to assume that zero is not in the spectrum $\{\lambda_k\}_{k=1}^\infty$ without loss of generality.
\end{rem}

\subsubsection{Notation for graph properties}
Given a Feynman ribbon graph $G$ we write
\begin{align*}
&n&&\text{for the number of external edges;}\\
&V&&\text{for the number of vertices ($V=\#G^0$);}\\
&E&&\text{for the number of internal edges (the propagators);}\\
&\Vfi&&\text{for the number of `fully internal vertices': vertices that border only unbroken faces;}\\
&\Efi&&\text{for the number of `fully internal edges': edges that border only unbroken faces;}\\
&F&&\text{for the number of faces;}\\
&B&&\text{for the number of broken faces;}\\
&U&&\text{for the number of unbroken faces ($U=F-B$);}\\
&L&&\text{for the loop order: $L:=1+E-V$;}\\
&g&&\text{for the genus ($2g=L+1-F$).}
\end{align*}

\section{Estimates on divided differences} 
\label{sct:estimates}
\subsection{Weighted divided differences}

Rather than estimating divided differences directly, it turns out to be more convenient to work with the following objects.

\begin{defn}
The \textbf{weighted divided difference} of $f'$ of order $n$ is defined as
\begin{align}\label{eq:wdivdif}
f'\{x_1,\ldots,x_n\}:=(-1)^{n} x_1\cdots x_nf'[x_1,\ldots,x_n].
\end{align}
\end{defn}

In particular, for $f'(x)=-x^{-p-1}$ with $p$ even, an explicit formula for the weighted divided difference is
$f'\{x_1,\ldots,x_n\}=\sum_{1\leq i_1\leq\cdots\leq i_{p}\leq n}x_{i_1}^{-1}\cdots x_{i_{p}}^{-1}$ by \cite[Theorem 15]{Dub20}.

For even functions with positive derivatives on the negative real line, all weighted divided differences are positive when evaluated sufficiently far away from zero.  This is due to the following sign property of their classical divided differences, which generalizes Hunter's positivity theorem~\cite{Hunter1977} concerning the positivity of the functions $\sum_{1\leq i_1\leq\cdots\leq i_{p}\leq n}x_{i_1}^{-1}\cdots x_{i_{p}}^{-1}$.

\begin{lem}[Signs of divided differences]\label{lem: divdiffs_nonzero}
Let $f\in C^{\infty}(\mathbb{R})$ be such that for all $n\in\mathbb{N}$ there exists~${R_n>0}$ such that $f^{(n)}(x)>0$ for all $x<-R_n$ and $(-1)^nf^{(n)}(x)>0$ for all $x>R_n$. 
Then for all $n$ and all~${\max_{k\leq n}R_k}<|x_1|\leq\cdots\leq|x_n|$
we have $f'[x_1,\dotsc,x_n]\neq 0$ and\begin{equation}\label{eq:DivDifSign}\mathrm{sgn}\;f'[x_1,\dotsc,x_{n}]=(-1)^n\mathrm{sgn}\;(x_1\cdots x_n).\;
\end{equation}
Equivalently,
\begin{align*}
f'\{x_1,\dotsc,x_n\}>0.
\end{align*}
\end{lem}
\begin{proof}
By induction. For $n=1$, let $|x|>R_1$. Then $f'[x]=f'(x)$ is nonzero by assumption. Furthermore, the assumptions on $f$ imply $\mathrm{sgn}\;f'(x)=-\mathrm{sgn}\;x$.

Now let $n>1$ and assume~\eqref{eq:DivDifSign} holds up to order $n-1$. 
Fix ${\max_{k\leq n}R_k} < |x_1|\leq \cdots \leq |x_n|$. If  $\sgn x_1=\cdots=\sgn x_n$, then by the intermediate value theorem there exists~$\xi\in\mathbb{R}$ in between $x_1$ and $x_n$ such that $$f'[x_1,\dotsc,x_{n}]=\frac{f^{(n)}(\xi)}{(n-1)!}\ne 0.$$ To show the sign condition, first let all $x_i$ be negative. By assumption, this implies that~$f^{(n)}(\xi)$ is positive, hence
\begin{equation*}
\mathrm{sgn}\;f'[x_1,\dotsc,x_{n}]=\mathrm{sgn}\;f^{(n)}(\xi)= 1 = (-1)^{n}\cdot (-1)^{n} = (-1)^{n}\cdot\mathrm{sgn}\;(x_1\cdots x_{n}).
\end{equation*}
Now let all $x_i$ be positive. We obtain
\begin{equation*}
\mathrm{sgn}\;f'[x_1,\dotsc,x_{n}]=\mathrm{sgn}\;f^{(n)}(\xi) =(-1)^{n} = (-1)^{n}\mathrm{sgn}\;(x_1\cdots x_{n}).
\end{equation*} This concludes the proof in the case where all $x_i$ have the same sign.

Now assume that there exists $x_k$ such that $\mathrm{sgn}\;x_k\neq \mathrm{sgn}\;x_{n}$ and consider  
\begin{equation*}
f'[x_1,\dotsc,x_{n}] = \frac{f'[x_1,\dotsc,x_{k-1},x_{k+1},\dotsc,x_{n}]-f'[x_1, \dotsc,x_{n-1}]}{x_{n}-x_k}.
\end{equation*}
Since ${\max_{k\leq n-1}R_{k} \leq \max_{k\leq n}R_k} < |x_1| \leq \cdots \leq |x_n|$, we can use the induction hypothesis to control the divided differences in the numerator. These have opposite signs, because
\begin{align*}
\mathrm{sgn}\;f'[x_1,\dotsc,x_{k-1},x_{k+1},\dotsc,x_{n}] &=(-1)^{n-1}\mathrm{sgn}\;x_{n}\cdot \mathrm{sgn}\;(x_1\cdots x_{k-1}x_{k+1}\cdots x_{n-1})\\ 
&= (-1)^{n-1} (-\mathrm{sgn}\;x_{k})\cdot \mathrm{sgn}\;(x_1\cdots x_{k-1}x_{k+1}\cdots x_{n-1}) \\
&=-\mathrm{sgn}\;f'[x_{1},\dotsc,x_{n-1}].
\end{align*} Since $x_{n}$ and $x_k$ have opposite signs, we further have $\mathrm{sgn}\;(x_{n}-x_k)=\mathrm{sgn}\;x_{n}$.     We may therefore write
\begin{multline*}
f'[x_1,\dotsc,x_{n}] = \\ -(\sgn f'[x_{1},\dotsc,x_{n-1}])\cdot(\sgn x_n)\frac{|f'[x_1,\dotsc,x_{k-1}, x_{k+1},\dotsc, x_{n}]|+|f'[ x_{1},\dotsc, x_{n-1}]|}{| x_{n}|+| x_k|}.
\end{multline*}
Since by induction we know that both divided differences in the numerator are nonzero, this implies that the numerator itself is also nonzero, and hence $f'[ x_1,\dotsc, x_{n+1}]\ne0$. The claimed sign property follows from the formula above, together with the induction hypothesis.
\end{proof}

The above observation allows us to express the amplitude \eqref{eqn: ampl_normal_divdiff} of a Feynman graph as a sum where, asymptotically, all terms have the same sign, regardless of the signs of the eigenvalues.
This dramatically limits cancellations, which will permit us to give a lower bound for the divergence.

\begin{lem}\label{lem:Ampl for weighted divdifs} 
Let $f$ be as in Lemma~\ref{lem: divdiffs_nonzero} such that $f'[\lambda_i,\lambda_j]\neq0$ for all $i,j\in\N$, and let $G$ be a Feynman ribbon graph. For all $N,i_1,\ldots,i_n\in\N$ with $\lambda_{i_j}\neq 0$, we have
$$\Ampl_{N,i_1,\ldots,i_n}(G)=(-1)^{\#\text{odd-degree vertices}} \lambda_{i_1}^{-1}\cdots \lambda_{i_n}^{-1}\sum_{i_{n+1},\ldots,i_{n+U}=1}^N\frac{\prod_{v\in G^0} f'\{i(\alpha^v_1),\ldots,i(\alpha^v_{\deg(v)})\}}{\prod_{e\in G^1_{\textnormal{int}}} f'\{i(\beta^e_1),i(\beta^e_2)\}},$$
where 
the numbers $\alpha^v_m,\beta^e_m\in\{1,\ldots,n+U\}$ and~$\deg(v)\in\N$ are determined by $G$ as in 
Section~\ref{sect: Feynman_rules}. 
\end{lem}
\begin{proof}
We have 
\begin{align*}
\Ampl_{N,i_1,\ldots,i_n}(G)&=\sum_{i_{n+1},\ldots,i_{n+U}=1}^N\frac{\prod_{v\in G^0} f'[i(\alpha^v_1),\ldots,i(\alpha^v_{\deg(v)})]}{\prod_{e\in G^1_{\textnormal{int}}} f'[i(\beta^e_1),i(\beta^e_2)]},
\end{align*}
 By definition, it follows that
$$\Ampl_{N,i_1,\ldots,i_n}(G)=\sum_{i_{n+1},\ldots,i_{n+U}=1}^N\frac{\prod_{v\in G^0} (-1)^{\deg (v)} \lambda_{i(\alpha^v_1)}^{-1}\cdots \lambda_{i(\alpha^v_{\deg(v)})}^{-1}f'\{i(\alpha^v_1),\ldots,i(\alpha^v_{\deg(v)})\}}{\prod_{e\in G^1_{\textnormal{int}}} \lambda_{i(\beta^e_1)}^{-1} \lambda_{i(\beta^e_2)}^{-1}f'\{i(\beta^e_1),i(\beta^e_2)\}}.$$
We refer to an index 
$i(\gamma)$ as \textit{running} if it corresponds to an unbroken face and as \textit{fixed} if it corresponds to a broken face.
Notice that:
\begin{itemize}
\item For each running index $i(\gamma)$, i.e.\ $\gamma\in\{n+1,\ldots,n+U\}$, there are precisely as many $ \lambda_{i(\gamma)}^{-1}$ in the numerator as in the denominator.
\item For each fixed index $i(\gamma)$ i.e.\ $\gamma\in\{1,\ldots,n\}$, $ \lambda_{i(\gamma)}^{-1}$ occurs precisely one time more often in the numerator than in the denominator.
\end{itemize}
These remarks together yield the lemma.
\end{proof}



\subsection{Bounds for weighted divided differences with sufficiently large inputs}\label{sect:bounds_large_k}
In order to estimate the amplitude of a Feynman graph as given in Lemma~\ref{lem:Ampl for weighted divdifs}, we need suitable estimates for the weighted divided differences. In this work, we are considering the following class of functions (cf. \cite[Definition 5]{HMvN2024}).
\begin{defn}\label{def:precise order}
    We say that $f\in C^\infty(\R)$ is \textbf{of precise order} $r\in\R$ if for all $k\in\N$ there exist positive numbers $R,c_1,c_2>0$ (depending on $k$), such that
    \begin{equation}\label{eqn: precise_order}
        c_1|x|^{r-k}\leq |f^{(k)}(x)|\leq c_2|x|^{r-k}
    \end{equation}
    for all $x\in\R\setminus[-R,R]$. If only the upper bound is assumed, we say $f$ is \textbf{of order at most} $r\in\R$.
    
\end{defn}

Note that for $f$ of precise order $r$, these conditions imply that all derivatives $f^{(k)}$ are nonzero outside a compact region (which may depend on $k$). Furthermore, for $r<0$ they imply that all derivatives of $f$ have the same sign on $\mathbb{R}_-$ sufficiently far away from zero. If $f$ additionally has the same sign on both $\mathbb{R}_+$ and $\mathbb{R}_-$ sufficiently far away from zero, for instance when $f$ is even, this implies the conditions of Lemma \ref{lem: divdiffs_nonzero} up to a global sign.

A key result is that we can express the weighted divided difference of~$f'$ as a divided difference of a function with polynomial-like behaviour as follows. 
\begin{lem}\label{lem: g_n}
Let $f\in C^{\infty}(\mathbb{R})$ such that $f'$ is of order at most $-1$ (cf.\ Definition \ref{def:precise order}). For $n\ge 2$, we define $g_n:\mathbb{R}\to \mathbb{R}$ by 
$$g_n(x):=-x^{n-2}f'(\tfrac{1}{x}),$$ where we set $g_n(0):=0$. The following holds.
\begin{enumerate}
\item We have $g_n\in C^{n-2}(\R)\cap C^{\infty}(\R\setminus\{0\})$ and, for all $x\neq0$, 
\begin{equation*}
g_n^{(n-1)}(x)=(-1)^n\frac{1}{x^n}f^{(n)}(\tfrac{1}{x}).
\end{equation*}
\item If $f'$ is of order at most $-p-1$ for some $p\in\R_{\geq0}$, then, as $x\downarrow0$, we have a constant $c_2>0$ (depending on $n$) such that
\begin{equation*}
|g_n^{(n-1)}(x)|\le c_2 |x|^p.
\end{equation*}
\item If $f'$ is of precise order $-p-1$ for some $p\in\R_{\geq0}$, then, as $x\downarrow0$, we have constants $c_1,c_2>0$ (depending on $n$) such that
\begin{equation*}
c_1|x|^p\le|g_n^{(n-1)}(x)|\le c_2 |x|^p.
\end{equation*}

\item If $f'$ is of order at most $r<-1$, then $g_n\in C^{n-1}(\R)$.

\item For all $x_1,\dotsc,x_n\neq 0$ we have
\begin{equation} \label{eq:DivDifasgn}
f'\{ x_1,\dotsc, x_n\}=g_n[ x_1^{-1},\dotsc, x_n^{-1}]. 
\end{equation}
\end{enumerate}
\end{lem}

\begin{rem} An example of an even function $f\in C^\infty(\R)$ with $f'$ of precise order~$-1$ but~${g_n\notin C^{n-1}(\R)}$ is obtained by choosing $f$ such that $f'(x)=\frac{3+\cos\log(|x|)}{x}$ on $\mathbb{R}\setminus[-1,1]$, the proof of which does not fit in a small remark.
\end{rem}

\begin{proof}[Proof of Lemma \ref{lem: g_n}.] 
We prove the first claim by induction. For $n=2$, we immediately see that \begin{equation*}
        g_2'(x)=\frac{1}{x^2}f^{(2)}(\tfrac{1}{x}).
    \end{equation*}
    Now note that \begin{equation*}
        g_{n+1}(x)=xg_n(x),
    \end{equation*}
    and hence \begin{equation*}
        g^{(n)}_{n+1}(x)=\frac{d^n}{dx^n}xg_n(x) = xg^{(n)}_n(x)+ng_n^{(n-1)}(x).
    \end{equation*}
    By induction hypothesis, the derivatives can be calculated explicitly and yield \begin{align*}
        xg^{(n)}_n(x)+ng_n^{(n-1)}(x) &= x\frac{d}{dx} \left((-1)^n\frac{1}{x^n}f^{(n)}(\tfrac{1}{x})\right) + n\left((-1)^n\frac{1}{x^n}f^{(n)}(\tfrac{1}{x})\right) \\
        &= (-1)^n \left( -n\frac{x}{x^{n+1}}f^{(n)}(\tfrac{1}{x}) + \frac{x}{x^n} \left(-\frac{1}{x^2}\right) f^{(n+1)}(\tfrac{1}{x}) + n\frac{1}{x^n} f^{(n)}(\tfrac{1}{x}) \right)\\ 
        &= (-1)^{n+1}\frac{1}{x^{n+1}}f^{(n+1)}(\tfrac{1}{x}).
    \end{align*}
    The fact that $g_n\in C^{n-2}(\mathbb{R})\cap C^{\infty}(\mathbb{R}\setminus\{0\})$ follows from $g_n(x)=x^n\frac{d}{dx}f(\tfrac{1}{x})$, where the order of $f'$ guarantees well-definedness of $g^{(k)}_n(0)$, $0\le k\le n-2$.

The second and third claim follow directly from the first statement. Note also that if $f'$ is of order at most $r<-1$, then the second claim implies that $g_n^{(n-1)}(0)=0$ and that $g_n^{(n-1)}$ is continuous at zero, proving the fourth claim.
   
Finally, to show~\eqref{eq:DivDifasgn},  we may assume that $ x_j \neq  x_k$ for all $j \neq k$, as the general case follows by continuity of the divided differences. Combining~\eqref{eqn: divdiff_written_out} and \eqref{eq:wdivdif}, we may write
\begin{equation*}
f'\{ x_1,\ldots, x_n\}
=(-1)^n  \sum_{j=1}^n\frac{x_1\cdots x_nf'( x_j)}{\prod_{k\neq j}( x_j- x_k)}.
\end{equation*} 
By using that $\frac{x_k}{ x_j- x_k}=\frac{x_j^{-1}}{ x_k^{-1}- x_j^{-1}}$, this implies
\begin{align*}
f'\{ x_1,\ldots, x_n\}
=(-1)^n \sum_{j=1}^n\frac{ x_j^{2-n}f'( x_j)}{\prod_{k\neq j} (x_{k}^{-1} - x_j^{-1})}
=\sum_{j=1}^n\frac{g_n(x_j^{-1})}{\prod_{k\neq j}( x_j^{-1}- x_{k}^{-1})},
    \end{align*}
    which yields equation \eqref{eq:DivDifasgn}.
\end{proof}

In order to prove our two-sided asymptotic estimates (Theorem \ref{thm:weighted_divdiff_bound}), we first assume that the inputs of the divided difference all lie on the positive half-line, before extending our estimates to inputs with arbitrary signs.
On the positive half-line, expressing weighted divided differences of~$f'$ as divided differences of $g_n$ allows us to obtain the desired bounds. 
\begin{prop}
\label{prop:weighted_divdiff_same_sign_bound} Let $n\in\N,$ $p\in\R_{\ge0}$, and let $f\in C^{\infty}(\mathbb{R})_\R$ be 
such that $f'$ is of precise order $-p-1$. Then there exist numbers $R,c_1,c_2\in\R_{>0}$ such that for all $x_1,\ldots,x_n\in\R$ satisfying $R< x_1\leq\cdots\leq x_n$ we have
$$c_1 x_1^{-p}\leq |f'\{x_1,\ldots,x_n\}|\leq c_2x_1^{-p}.$$
Indeed, one may use the number $R$ from Definition \ref{def:precise order}. {If $f'$ is only assumed to be of order at most $-p-1$ then the upper bound still holds.}
\end{prop}
\begin{proof} 
  By Lemma~\ref{lem: g_n}, this proof reduces to providing an upper and lower bound on~$g_n[ x_1^{-1},\ldots,  x_n^{-1}]$. Note that by construction, $g_n^{(n-1)}$ has the same sign on $(0,\tfrac{1}{C})$ as the derivatives of $f$ on $\mathbb{R}_-$ for sufficiently large $C>0$. Assume that the derivatives of $f$ are positive on $(-\infty,-R)$; the case where they are negative follows similarly. The upper bound is a straightforward consequence of Lemma~\ref{lem: g_n} 
and the intermediate value theorem, since for some $\xi\in[ x_n^{-1}, x_1^{-1}]$ we have \begin{equation}\label{eq:gnupper}
    g_n[ x_1^{-1},\dotsc, x_n^{-1}]= \frac{1}{(n-1)!}g_n^{(n-1)}(\xi)\lesssim \xi^p\le  x_1^{-p}.
\end{equation}

Towards a lower bound we again assume that $ x_j \neq  x_k$ for all $j \neq k$, the general case then follows by continuity of the divided differences. We use the integral expression in~\eqref{eqn: divdiff_integral_form}, yielding 
    \begin{multline*}
        g_n[ x_1^{-1},\ldots, x_n^{-1}]\\= \int_0^1dt_{n-1}\int_0^{t_{n-1}}dt_{n-2}\cdots\int_0^{t_2}dt_1 g_n^{(n-1)}( x_{n}^{-1}+( x_{n-1}^{-1}- x_n^{-1})t_{n-1}+\cdots+( x_1^{-1}- x_2^{-1})t_1).
    \end{multline*}
    We substitute $r_1:= x_{n}^{-1}+( x_{n-1}^{-1}- x_n^{-1})t_{n-1}+\cdots+( x_1^{-1}- x_2^{-1})t_1$ in the innermost integral to obtain \begin{equation*}
         \int_0^1dt_{n-1}\int_0^{t_{n-1}}dt_{n-2}\cdots\int_0^{t_3}dt_2\frac{1}{ x_1^{-1}- x_2^{-1}}\int_{\tau_2}^{\tau_1}dr_1 g_n^{(n-1)}(r_1),
    \end{equation*}
    where for $i=1,2$ we define \begin{align*}
        \tau_i&=  x_{n}^{-1}+( x_{n-1}^{-1}- x_n^{-1})t_{n-1}+\cdots+( x_i^{-1}- x_3^{-1})t_2.
    \end{align*}
    We use Lemma~\ref{lem: g_n} to obtain the lower bound \begin{multline}\label{eq:MultiIntegral}
          \int_0^1dt_{n-1}\int_0^{t_{n-1}}dt_{n-2}\cdots\int_0^{t_3}dt_2\frac{1}{ x_1^{-1}- x_2^{-1}}\int_{\tau_2}^{\tau_1}dr_1 g_n^{(n-1)}(r_1) \\ \gtrsim   \int_0^1dt_{n-1}\int_0^{t_{n-1}}dt_{n-2}\cdots\int_0^{t_3}dt_2\frac{1}{ x_1^{-1}- x_2^{-1}}\int_{\tau_2}^{\tau_1}dr_1 r_1^{p}.
    \end{multline}

    First, note that we may rewrite $\tau_i$ as
    \[
    \tau_i=  x_i^{-1}t_2 +  x_3^{-1}(t_3-t_2)+\dotsc+ x_{n-1}^{-1}(t_{n-1}-t_{n-2})+ x_n^{-1}(1-t_{n-1}),
    \]
    where by construction all summands are positive. Since
    $    \tau_1 - \tau_2 = t_2( x_1^{-1} -  x_2^{-1})$,
    it follows that
    \begin{align*}
        \frac{1}{ x_1^{-1}- x_2^{-1}}\int_{\tau_2}^{\tau_1}dr_1 r_1^{p} &= \frac{t_2}{p+1}\frac{1}{\tau_1 -\tau_2}   (\tau_1^{p+1}-\tau_2^{p+1})\\
         &= \frac{t_2}{p+1} \tau_1^{p} \frac{1-\big(\frac{\tau_2}{\tau_1}\big)^{p+1}}{1-\big(\frac{\tau_2}{\tau_1}\big)}.        
    \end{align*}
    Since
    \[
    \frac{1-r^{p+1}}{1-r} \geq 1, \quad  p\ge0,\, r\in(0,1),
    \]
    we therefore have that
    \[
    \frac{1}{ x_1^{-1}- x_2^{-1}}\int_{\tau_2}^{\tau_1}dr_1 r_1^{p} \geq \frac{t_2}{p+1} \tau_1^p \geq \frac{t_2^{p+1}}{p+1}  x_1^{-p}.
    \]
    Since the remaining integrals in~\eqref{eq:MultiIntegral} are independent of $ x_1$, we thus find that
    \begin{equation}\label{eq:gnlower}
    g_n[ x_1^{-1},\ldots,  x_n^{-1}] \gtrsim  x_1^{-p}. 
    \end{equation}
    Combining equations~\eqref{eq:DivDifasgn},~\eqref{eq:gnupper} and~\eqref{eq:gnlower} finishes the proof.
\end{proof}

By induction, we can now show the same bound while allowing for inputs with mixed signs.
\begin{thm}
\label{thm:weighted_divdiff_bound} Let $n\in\N_{\geq1}$ and $p\in\R_{\geq0}$ and let $f\in C^{\infty}(\mathbb{R})_{\R}$  
be such that $f'$ is of precise order $-p-1$ and $f$ is positive outside a compact region.
Then there exist $R,c_1,c_2\in\R_{>0}$ such that for all $ x_1,\ldots, x_n\in\R$ satisfying $R<| x_1|\leq\cdots\leq| x_n|$ we have
$$c_1| x_1|^{-p}\leq|f'\{ x_1,\ldots, x_n\}|\leq c_2| x_1|^{-p}.$$
The upper bound remains valid for any $f\in C^\infty(\R)$ such that $f'$ is of order at most $-p-1$.
\end{thm}
\begin{proof}
By induction. The cases $n=1$ and $\sgn x_1=\dotsc=\sgn x_n =1$, $n>1$, immediately follow from Proposition~\ref{prop:weighted_divdiff_same_sign_bound}. The case $\sgn x_1=\dotsc=\sgn x_n =-1$, $n>1$, similarly follows from applying Proposition~\ref{prop:weighted_divdiff_same_sign_bound} to the function $h(x)=f(-x)$.

Hence we assume $n>1$ and that there exists $k\ne n$ such that $\sgn x_k\neq\sgn x_n$. 
This implies that $|x_n-x_k|=|x_n|+|x_k|$, and so by the triangle inequality we obtain
\begin{align}
    |f'\{ x_1,\dotsc, x_n\}|&\leq| x_1\cdots x_n| \frac{|f'[ x_1,\dotsc, x_{k-1}, x_{k+1},\dotsc, x_n]|+|f'[ x_1,\dotsc, x_{n-1}]|}{| x_n|+| x_k|}.
\end{align}
Since for $a,b,c,d >0$ we have
\[
\frac{a+b}{c+d} \leq \frac{a}{c} + \frac{b}{d},
\]
it follows that
\begin{align*}
    |f'\{ x_1,\dotsc, x_n\}|
    &\le\frac{| x_1\cdots x_n|}{| x_k|}|f'[ x_1,\dotsc, x_{k-1}, x_{k+1},\dotsc, x_n]| + \frac{| x_1\cdots x_n|}{| x_n|}|f'[ x_1,\dotsc, x_{n-1}]|  \\ 
    &= |f'\{ x_1,\dotsc, x_{k-1}, x_{k+1},\dotsc, x_n\}| +  |f'\{ x_1,\dotsc,x_{n-1}\}|.
\end{align*} 
The upper bound now follows by induction, where in the case $k=1$ we use that $| x_j|^{-p} \le | x_1|^{-p}$ for all $2\le j\le n$.

For the lower bound, we moreover assume that $f$ is positive outside a compact region. As in the proof of Lemma~\ref{lem: divdiffs_nonzero}, we obtain
\begin{align}\label{eqn: weighted_divdiff_decomp}
    |f'\{ x_1,\dotsc, x_n\}|&=| x_1\cdots x_n| \frac{|f'[ x_1,\dotsc, x_{k-1}, x_{k+1},\dotsc, x_n]|+|f'[ x_1,\dotsc, x_{k},\dotsc, x_{n-1}]|}{| x_n|+| x_k|}.
\end{align}
We use $| x_n|+| x_k|\le 2| x_n|$ and~\eqref{eqn: weighted_divdiff_decomp} to find that
\begin{align*}
    |f'\{ x_1,\dotsc, x_n\}|&\ge\frac{| x_1\cdots x_n| }{| x_n|+| x_k|} |f'[ x_1,\dotsc, x_{k},\dotsc, x_{n-1}]| \\
    &\geq  \frac{| x_1\cdots x_n|}{2| x_n|}|f'[ x_1,\dotsc, x_{k},\dotsc, x_{n-1}]| \\
    &= \frac{1}{2}| x_1\cdots x_{n-1}||f'[ x_1,\dotsc, x_{n-1}]| \\
    &=\frac{1}{2}|f'\{ x_1,\dotsc, x_{n-1}\}|.
\end{align*}
The lower bound now follows by induction.
\end{proof}

\begin{cor}\label{cor:wdivdifs are bounded}
For any $g\in C^\infty(\R)$ of order at most $-1$ and for any $n\in\N$,  the weighted divided difference $(x_1,\ldots,x_n)\mapsto g\{x_1,\ldots,x_n\}$ is a bounded function on $\R^n$.
\end{cor}
We note the independent interest of the above corollary, which we have not seen in the literature, but shares similarity with \cite[Definition 5 and Lemma 4.6]{HMvN2024} and \cite[(23)]{vNS22}.

\subsection{Bounds for weighted divided differences with singular inputs}
We here note the deviating behavior of the asymptotics of (weighted) divided differences when an input $\lambda_0$ satisfies $f'(\lambda_0)=0$. Recall that by Remark~\ref{rem:translational invariance}, we may assume $\lambda_0\neq 0$ without loss of generality.
\begin{lem}\label{lem:zero-index bounds}
Let $f\in C^\infty(\R)_\R$ be 
such that $f'$ is of precise order $-p-1$ for some $p\in\R_{\geq0}$.
Let $\lambda_0\in\R\setminus\{0\}$ be such that $f'(\lambda_0)=0$, and let $n\in\N_{\geq1}$. Then there exist $R,c_1,c_2\in\R_{>0}$ 
such that for all $x\in\R\setminus[-R,R]$ we have
$$c_1|x|^{-p-1}\leq |f'\{\lambda_0,x\}|$$
and for all $x_1,\ldots,x_n\in\R$ satisfying $R<|x_1|\leq\cdots\leq|x_n|$ we have
$$|f'\{\lambda_0,x_1,\ldots,x_n\}|\leq c_2|x_1|^{-p-1}.$$
\end{lem}
\begin{proof}
The lower bound follows directly from 
$$f'\{\lambda_0,x\}=\lambda_0x\frac{f'(x)}{x-\lambda_0}.$$
To obtain the upper bound, we set
$$F(x):=\frac{f'(x)}{x-\lambda_0}\quad (x\in\R\setminus\{\lambda_0\})$$
and $F(\lambda_0):=f^{(2)}(\lambda_0)$,
which implies, e.g.\ by~\eqref{eqn: divdiff_written_out},
$$f'\{\lambda_0,x_1,\ldots,x_n\}=-\lambda_0F\{x_1,\ldots,x_n\}.$$
The function $F\in C^\infty(\R)$ has precise order $-p-2$ and, although it is not the derivative of 
a function positive outside a compact, we may apply the last statement of Theorem \ref{thm:weighted_divdiff_bound}, which yields the desired upper bound.
\end{proof}


\subsection{Bounds for weighted divided differences for all inputs}
The aim of this section is to prove a two-sided bound for the (weighted) divided difference which works for all inputs, not necessarily those outside a compact region. Of course, such a lower bound can only be established if the divided differences do not vanish identically.

The first step towards this two-sided bound is the following remarkable property of the weighted divided difference. Although it illuminates divided differences themselves, we could not find this result in the literature (noting \cite[Lemma A.4]{PSST2017}). 
\begin{lem}\label{lem:limit}
Assume $f'\in C^{\infty}(\R)$ has precise order $-p-1$ for some $p\ge0$. For all $n,m\in\N_{\geq1}$ and all~$x_1,\ldots,x_n\in\R$ we have
$$\lim_{|y_1|,\ldots,|y_m|\to\infty}f'\{x_1,\ldots,x_n,y_1,\ldots,y_m\}=f'\{x_1,\ldots,x_n\}.$$
\end{lem}
\begin{proof}
In case $x_i = 0$ for any $i$, the statement is trivial. Hence, assume that $x_1,\ldots,x_n \in \R\setminus\{0\}$.

Recall the functions $g_n$ introduced in Lemma~\ref{lem: g_n} and note $g_{n}(0)=\lim_{x\to0}(-x^{n-2}f'(\frac1x))=0$ for all $n\ge 2$. Define auxiliary maps $g_{n,k}$ inductively in~${k\in\N}$ as $g_{n,0}:=g_n$ and $g_{n,k}:=g_{n,k-1}[\cdot,0]$. Then Lemma~\ref{lem: g_n} and~\cite[Lemma 3.4]{CaspersSukochevZanin2021} (which notably only requires  $g_{n+m}\in C^{n+m-2}(\mathbb{R})\cap C^{n+m-1}(\mathbb{R}\setminus\{0\})$, which holds by the first part of Lemma~\ref{lem: g_n})
yield \begin{equation}\label{eq:gn+m_eq_1}
    g_{n+m}[x_1,\ldots,x_n,0,\ldots,0]=g_{n+m,m}[x_1,\dotsc,x_n].
\end{equation}
We show that $g_{n+m,m}=g_n$ by induction on $m$. Indeed, if $m=1$ then \begin{equation*}
    g_{n+1,1}(x)=g_{n+1,0}[x,0]=\frac{g_{n+1}(x)-g_{n+1}(0)}{x}=x^{-1}g_{n+1}(x)=g_n(x).
\end{equation*}
Now assume the statement holds for all $n\in\mathbb{N}_{\ge 1}$ and up to some $m\in\N_{\ge 1}$. Then \begin{equation*}
   g_{n+m+1,m+1}(x)=g_{n+m+1,m}[x,0]=g_{(n+1)+m,m}[x,0]=g_{n+1}[x,0]=g_{n+1,1}(x)=g_n(x), 
\end{equation*}
where the induction hypothesis was used twice.

Next, we show that for $x_1,\ldots,x_n \neq 0$,  
\begin{equation}\label{eq:gn+m_eq_2}
    \lim_{t_1,\ldots,t_m\to 0} g_{n+m}[x_1,\ldots,x_n, t_1,\ldots,t_m] = g_{n+m}[x_1,\ldots,x_n,0,\ldots,0].
\end{equation}
By Lemma~\ref{lem: g_n} we have that $g_{n+m}\in C^{n+m-2}(\R)$, which implies that $g_{n+m}^{[n+m-2]}\in C(\R^{n+m-1})$. Therefore,
\begin{align*}
    \lim_{t_1,\ldots,t_m\to 0}& g_{n+m}^{[n+m-1]}[x_1,\ldots,x_n, t_1,\ldots,t_m] \\
    &= \lim_{t_1,\ldots,t_m\to 0} \frac{g_{n+m}^{[n+m-2]}[x_1,\ldots,x_n,t_1,\ldots, t_{m-1}]-g_{n+m}^{[n+m-2]}[x_2,\ldots,x_n,t_1,\ldots, t_{m}]}{x_1-t_m}\\
    &= \frac{g_{n+m}^{[n+m-2]}[x_1,\ldots,x_n,0,\ldots,0]-g_{n+m}^{[n+m-2]}[x_2,\ldots,x_n,0,\ldots,0]}{x_1}\\
    &= g_{n+m}^{[n+m-1]}[x_1,\ldots,x_n,0,\ldots,0].
\end{align*}

From equations \eqref{eq:DivDifasgn},~\eqref{eq:gn+m_eq_1},~\eqref{eq:gn+m_eq_2} and the identity $g_{n+m,m}=g_n$, it follows that for $x_1,\ldots,x_n \neq 0$,
$$f'\{x_1,\ldots,x_n,y_1,\ldots,y_n\}=g_{n+m}[x_1^{-1},\ldots,x_n^{-1},y_1^{-1},\ldots,y_m^{-1}]\to g_n[x_1^{-1},\ldots,x_n^{-1}]=f'\{x_1,\ldots,x_n\},$$
as $|y_j|\to\infty$.
\end{proof}
\begin{lem}\label{lem:weighted_divdifs_bounds_small-large}
Assume $f'\in C^\infty(\R)$ has precise order $-p-1$ for some $p\ge0$. For all $n\in\N_{\geq1}$, for all~$x_1,\ldots,x_n\in\R$ that satisfy $f'\{x_1,\ldots,x_n\}\neq0$, and for all $m\in\N$, there exist $C_1,C_2>0$ and some~$R_{m,x_1,\ldots,x_n}\equiv R_x$ such that, for all $y_1,\ldots,y_m\in\R\setminus[-R_x,R_x]$, we have
\begin{align}\label{eq:wdivdif_smalllarge_constants}
C_1\leq|f'\{x_1,\ldots,x_n,y_1,\ldots,y_m\}|\leq C_2.
\end{align}
\end{lem}
\begin{proof}
Follows immediately from Lemma \ref{lem:limit}.
\end{proof}

In order to get a clean lower bound for the weighted divided difference, uniform in its arguments, it is convenient to assume that zero is not in the spectrum. We note that this can be done without loss of generality, by a simple shifting argument: Because the spectrum is discrete, there exists $\epsilon\in\R$ such that $-\epsilon\notin\cup_{k=1}^\infty\{\lambda_k\}$. By Remark \ref{rem:translational invariance}, it suffices to prove the main theorem for~$f^\epsilon,\lambda^\epsilon$.

\begin{thm}\label{thm:no_vanishing_divdifs}
Let $f\in C^\infty(\R)$ be 
such that $f'$ is of precise order $-p-1$ for some $p\ge0$, and $f$ is positive outside a compact region.
Let $n\in\N_{\geq1}$ and assume that $\{\lambda_k\}_{k=1}^\infty\subseteq\R\setminus\{0\}$ is such that $|\lambda_k|\to\infty$ as $k\to \infty$ and such that $f'[\lambda_{k_1},\ldots,\lambda_{k_m}]\neq0$ for every $m=1,\ldots,n$
and $k_1,\ldots,k_m\in\N_{\geq1}$. There exist $c_1,c_2>0$ such that, for all $k_1,\ldots,k_n$ with $|\lambda_{k_1}|\leq\cdots\leq|\lambda_{k_n}|$, we have
\begin{align}\label{eq:wdivdif two bounds}
c_1|\lambda_{k_1}|^{-p}\leq|f'\{\lambda_{k_1},\ldots,\lambda_{k_n}\}|\leq c_2 |\lambda_{k_1}|^{-p}.
\end{align}
\end{thm}
\begin{proof}
Let $R$ be such that the defining equation of the precise order of $f'$ in Definition~\ref{def:precise order} holds for all derivatives up to order $n$ simultaneously on $\R\setminus[-R,R]$. With $k_1$ ranging over all values such that $|\lambda_{k_1}|\leq R$, Lemma \ref{lem:weighted_divdifs_bounds_small-large} yields an associated $R_2^{k_1}$. Define $R_2:=\max_{k_1,|\lambda_{k_1}|\leq R} R_2^{k_1}$. Given $k_1,k_2$, Lemma \ref{lem:weighted_divdifs_bounds_small-large} yields an associated $R_3^{k_1,k_2}$. Define $R_3:=\max_{k_1,k_2,|\lambda_{k_1}|\leq R,|\lambda_{k_2}|\leq R_2}R_3^{k_1,k_2}.$ Repeating this procedure yields $R_2,\ldots,R_n$.

We are now ready to prove \eqref{eq:wdivdif two bounds} with a case distinction over $n+1$ classes of indices $k_1,\ldots,k_n$.
First, suppose that 
$$R<|\lambda_{k_1}|\leq \cdots\leq|\lambda_{k_n}|.$$ 
For this class of indices, Theorem \ref{thm:weighted_divdiff_bound} directly implies \eqref{eq:wdivdif two bounds}.
Next, suppose that
$$|\lambda_{k_1}|\leq R,\quad R_2\leq|\lambda_{k_2}|\leq\cdots\leq|\lambda_{k_n}|.$$
In this case, by construction of $R_2$ (using Lemma \ref{lem:weighted_divdifs_bounds_small-large}) we have $C_1=C_1^{k_1}>0,C_2=C_2^{k_1}>0$ such that \eqref{eq:wdivdif_smalllarge_constants} holds. Defining $c_1^{k_1}:=C_1^{k_1}|\lambda_{k_1}|^p$ and $c_2^{k_1}:=C_2^{k_1}|\lambda_{k_1}|^p$, we obtain \eqref{eq:wdivdif two bounds} with constants that depend on $k_1$. But there are a finite amount of those, so after taking $c_1=\min_{k_1:|\lambda_{k_1}|\leq R}c^{k_1}_1$ and $c_2=\max_{k_1:|\lambda_{k_1}|\leq R}c^{k_1}_2$ we obtain \eqref{eq:wdivdif two bounds} for this class of indices.
Now, suppose that
$$|\lambda_{k_1}|\leq R,\quad |\lambda_{k_2}|\leq R_2,\quad R_3\leq|\lambda_{k_3}|\leq\cdots\leq|\lambda_{k_n}|.$$
This case is similar to the previous one, but now we need to take $c_1=\min_{k_1,k_2:|\lambda_{k_1}|\leq R,|\lambda_{k_2}|\leq R_2} c_1^{k_1,k_2}$ and $c_2=\max_{k_1,k_2:|\lambda_{k_1}|\leq R,|\lambda_{k_2}|\leq R_2} c_2^{k_1,k_2}$. Similarly, we obtain \eqref{eq:wdivdif two bounds}. In fact, all cases up to 
$$|\lambda_{k_1}|\leq R,\quad\ldots,\quad |\lambda_{k_{n-1}}|\leq R_{n-1},\quad R_n<|\lambda_{k_n}|$$
proceed similarly. The final class of indices, satisfying 
$$|\lambda_{k_1}|\leq R,\quad\ldots,\quad |\lambda_{k_{n}}|\leq R_{n},$$
contains but a finite amount of indices. We automatically have \eqref{eq:wdivdif two bounds} for $k_1,\ldots,k_n$-dependent constants, as the weighted divided differences are nonzero. Again, taking the minimum over $c_1^{k_1,\ldots,k_n}$ and the maximum over $c_2^{k_1,\ldots,k_n}$, and finally taking the minimum/maximum over all $n+1$ classes, the proof is complete.
\end{proof}

\section{Two-sided power counting}\label{sect:proof_power_counting}
In the following two subsections, we will prove our first main result, which is the following.

\begin{thm}\label{thm:main}
Let $f\in C^{\infty}(\mathbb{R})_\R$ be an even function such that $f'$ is of precise order $-p-1$ for some~${p \in\mathbb{R}_{\ge0}}$.
Let $\{\lambda_k\}_{k=1}^\infty$ be a sequence of real numbers for which there exist constants $K,c_1,c_2>0$ such that for all~${k\geq K}$ we have~${c_1k^{1/d}\leq |\lambda_k|\leq c_2k^{1/d}}$. Assume $f'[\lambda_{k_1},\ldots,\lambda_{k_n}]\neq0$ for all ${k_1,\ldots,k_n\in\N_{\geq1}}$ and all~${n\in\N_{\geq1}}$. Then for any Feynman ribbon graph $G=(G^0,n,G^1)$ whose vertices have valence $\geq3$, for all external indices~${i_1,\ldots,i_n\in\N_{\geq1}}$, there exist $M,c_3,c_4>0$ such that, for all $N\geq M$,
$$c_3 N^{\omega(G)}\leq|\Ampl_{N,i_1,\ldots,i_n}(G)|\leq c_4 N^{\omega(G)};$$
\begin{align}\label{eq:omg}
\omega(G)=U+\frac{p}{d}(E_{\textnormal{fi}}-V_{\textnormal{fi}}),
\end{align}
where $U$ is the number of unbroken faces of $G$, $E_{\textnormal{fi}}$ is the number of fully internal edges of $G$ (propagators bordered on both sides by unbroken faces) and $V_{\textnormal{fi}}$ is the number of fully internal vertices of $G$ (vertices bordered on all sides by unbroken faces). We recall that $\Ampl_{N,i_1,\ldots,i_n}(G)$ is defined in~\eqref{eqn: ampl_normal_divdiff}.
\end{thm}

\subsection{Upper bound}\label{sct: upper_bound}
\begin{figure}[H]
\centering
\begin{subfigure}[t]{0.23\textwidth}
\centering
\begin{tikzpicture}[thick]
\filldraw[color=blue!20!white] (0,0) circle (1.65cm);
\draw (0,1.65) to (0,2);
\draw (0,-1.65) to (0,-2);
\draw (0,1.65) arc[start angle=90,end angle=-270,radius=1.65cm];
\draw[color=blue] (0,1.65) to (0,0.7);
\draw[color=blue] (.75,0) arc[start angle=0, end angle=360, radius=.75cm];
\draw[color=blue] (-.75,0) to (0,0);
\draw[color=blue] (.75,0) to (1.35,0);
\draw[color=blue] (0,-.75) to (0,-1.75);
{\color{blue}\vertex{.75,0};
\vertex{1.35,0};
\vertex{-.75,0};
\vertex{0,0};
\vertex{0,0.75};
\vertex{0,-0.75};}
\vertex{0,1.65};
\vertex{0,-1.65};
\node at (.9,-.8) {\scriptsize$1$};
\node at (-.9,-.8) {\scriptsize$2$};
\node at (.3,-.3) {\scriptsize$3$};
\end{tikzpicture}
\caption{All unbroken faces blue. Automatically, the blue vertices and edges are precisely the fully internal vertices and edges.}
\end{subfigure}
~~
\begin{subfigure}[t]{0.23\textwidth}
\centering
\begin{tikzpicture}[thick]
\filldraw[color=blue!20!white] (0,0) circle (1.65cm);
\filldraw[color=white] (0,1.65) rectangle (1.65,-1.65);
\filldraw[color=blue!20!white] (0,0) circle (.75cm);
\draw (0,1.65) to (0,2);
\draw (0,-1.65) to (0,-2);
\draw (0,1.65) arc[start angle=90,end angle=-270,radius=1.65cm];
\draw (0,1.65) to (0,0.7);
\draw[color=blue] (0,-.75) arc[start angle=270, end angle=90, radius=.75cm];
\draw (0,-.75) arc[start angle=-90, end angle=90, radius=.75cm];
\draw[color=blue] (-.75,0) to (0,0);
\draw (.75,0) to (1.35,0);
\draw (0,-.75) to (0,-1.75);
{\color{blue}
\vertex{-.75,0};
\vertex{0,0};
}
\vertex{0,-0.75};
\vertex{0,0.75};
\vertex{.75,0};
\vertex{1.35,0};
\vertex{0,1.65};
\vertex{0,-1.65};
\node at (1.2,.15) {$\curvearrowleft$};
\node at (.825,.225) {\rotatebox[origin=c]{-70}{$\curvearrowleft$}};
\node at (.125,.955) {\rotatebox[origin=c]{-90}{$\curvearrowleft$}};
\node at (.22,-.83) {\rotatebox[origin=c]{-160}{$\curvearrowleft$}};
\node at (.9,-.8) {\scriptsize$1$};
\node at (-.9,-.8) {\scriptsize$2$};
\node at (.3,-.3) {\scriptsize$3$};
\end{tikzpicture}
\caption{Two steps from induction base.}
\end{subfigure}
~~
\begin{subfigure}[t]{0.23\textwidth}
\centering
\begin{tikzpicture}[thick]
\filldraw[color=blue!20!white] (0,0) circle (.75cm);
\draw (0,1.65) to (0,2);
\draw (0,-1.65) to (0,-2);
\draw (0,1.65) arc[start angle=90,end angle=-270,radius=1.65cm];
\draw (0,1.65) to (0,0.7);
\draw (.75,0) arc[start angle=0, end angle=360, radius=.75cm];
\draw[color=blue] (-.75,0) to (0,0);
\draw (.75,0) to (1.35,0);
\draw (0,-.75) to (0,-1.75);
\vertex{.75,0};
\vertex{1.35,0};
\vertex{-.75,0};
{\color{blue}
\vertex{0,0};}
\vertex{0,0.75};
\vertex{0,-0.75};
\vertex{0,1.65};
\vertex{0,-1.65};
\node at (1.2,.15) {$\curvearrowleft$};
\node at (.825,.225) {\rotatebox[origin=c]{-70}{$\curvearrowleft$}};
\node at (.125,.955) {\rotatebox[origin=c]{-90}{$\curvearrowleft$}};
\node at (.22,-.83) {\rotatebox[origin=c]{-160}{$\curvearrowleft$}};
\node at (-.825,.225) {\rotatebox[origin=c]{70}{$\curvearrowright$}};
\node at (.9,-.8) {\scriptsize$1$};
\node at (-.9,-.8) {\scriptsize$2$};
\node at (.3,-.3) {\scriptsize$3$};
\end{tikzpicture}
\caption{One step from induction base.}
\end{subfigure}
~~
\begin{subfigure}[t]{0.23\textwidth}
\centering
\begin{tikzpicture}[thick]
\draw (0,1.65) to (0,2);
\draw (0,-1.65) to (0,-2);
\draw (0,1.65) arc[start angle=90,end angle=-270,radius=1.65cm];
\draw (0,1.65) to (0,0.7);
\draw (.75,0) arc[start angle=0, end angle=360, radius=.75cm];
\draw (-.75,0) to (0,0);
\draw (.75,0) to (1.35,0);
\draw (0,-.75) to (0,-1.75);
\vertex{.75,0};
\vertex{1.35,0};
\vertex{-.75,0};
\vertex{0,0};
\vertex{0,0.75};
\vertex{0,-0.75};
\vertex{0,1.65};
\vertex{0,-1.65};
\node at (1.2,.15) {$\curvearrowleft$};
\node at (.825,.225) {\rotatebox[origin=c]{-70}{$\curvearrowleft$}};
\node at (.125,.955) {\rotatebox[origin=c]{-90}{$\curvearrowleft$}};
\node at (.22,-.83) {\rotatebox[origin=c]{-160}{$\curvearrowleft$}};
\node at (-.825,.225) {\rotatebox[origin=c]{70}{$\curvearrowright$}};
\node at (-.2,.15) {$\curvearrowleft$};
\node at (.9,-.8) {\scriptsize$1$};
\node at (-.9,-.8) {\scriptsize$2$};
\node at (.3,-.3) {\scriptsize$3$};
\end{tikzpicture}
\caption{Induction base.}
\end{subfigure}
\caption{The algorithm of Lemma \ref{lem:injection fi vertices to fi edges}.}
\label{fig:blue}
\end{figure}
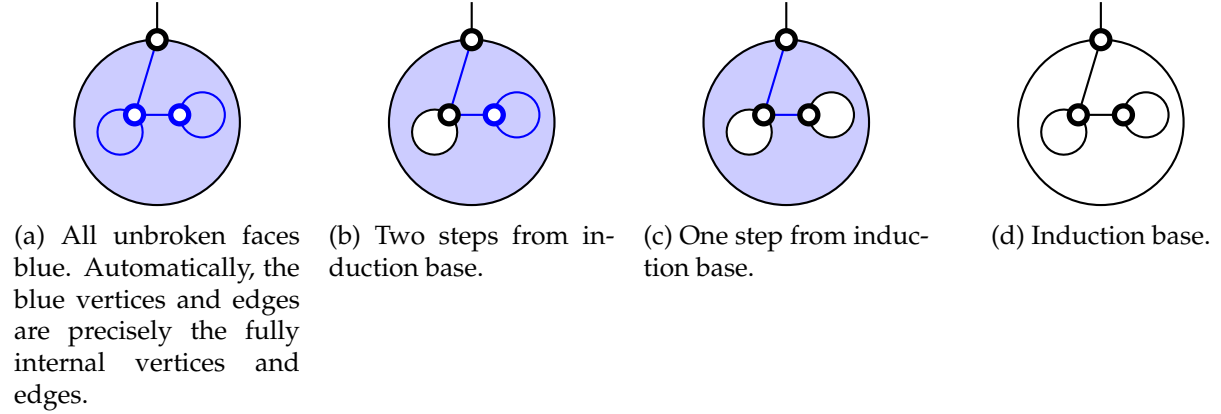
\begin{lem}\label{lem:injection fi vertices to fi edges}
Let $G$ be a Feynman ribbon graph. For any given total order on the set of unbroken faces, there exists an injective map $\Gamma_G$ from the fully internal vertices to the fully internal edges such that, for each fully internal $v\in G^0$,
\begin{itemize}
\item $\Gamma_G(v)$ is incident to $v$;
\item $\Gamma_G(v)$ borders a face which borders $v$ and is the minimum of the faces which border $v$.
\end{itemize}
\end{lem}
\begin{proof}
One may color any subset of the unbroken faces of $G$ blue, and color a vertex or edge blue if all faces that it borders are blue. For all such colorings, we shall prove the generalization of the lemma in which `fully internal' is replaced by `blue'. The proof proceeds by induction to the number of blue faces, cf.\ Figure \ref{fig:blue}. The induction base (no blue faces) is trivial. If $G$ has blue faces, we let $\mathcal{F}$ be the minimum of the blue faces. Now assign to each blue vertex $v$ bordering~$\mathcal{F}$ a neighboring blue edge $e=:\Gamma_G(v)$ as follows. All edges bordering $\mathcal{F}$ that attach to a 1-vertex are assigned to that vertex. All other edges that border $\mathcal{F}$ on two sides are ignored, and the remaining edges bordering $\mathcal{F}$ are organized in cycles, so that one may simply follow them clockwise (or anti-clockwise), finding that each new blue vertex is always followed by a new blue edge. We note that this assignment satisfies the conditions of the lemma. All other blue vertices can also be assigned according to the conditions of the lemma, by applying the induction hypothesis, after removing the blue color from $\mathcal{F}$ and from its bordering vertices and edges.
\end{proof}
\begin{thm}\label{thm:upper}
Assume the conditions of Theorem \ref{thm:main}.
Asymptotically in $N\in\N$, we have the following upper bound of the amplitude:
\begin{align*}
|\Ampl_{N,i_1,\ldots,i_n}(G)|&\lesssim N^{U+\frac{p}{d}(E_\textnormal{fi}-V_\textnormal{fi})}.
\end{align*}
\end{thm}
\begin{proof}
Throughout the proof we fix $i_1,\ldots,i_n$, and so we trivially have $|\lambda_{i_1}|,\ldots,|\lambda_{i_n}|\lesssim 1$ as $N\to\infty$. By Lemma \ref{lem:Ampl for weighted divdifs},
\begin{align}\label{eq:first bound amplitude}
|\Ampl_{N,i_1,\ldots,i_n}(G)|&\lesssim\sum_{i_{n+1},\ldots,i_{n+U}=1}^N\frac{\prod_{v\in G^0} |f'\{i(\alpha^v_1),\ldots,i(\alpha^v_{\deg(v)})\}|}{\prod_{e\in G^1_{\textnormal{int}}} |f'\{i(\beta^e_1),i(\beta^e_2)\}|},
\end{align}
where $i(\alpha)\equiv i_\alpha$, see~\eqref{eqn: ampl_normal_divdiff}.

For each choice of $i_{n+1},\ldots,i_{n+U}\in\{1,\ldots,N\}$, we equip the faces of $G$ with a total order that corresponds to $i_1,\ldots,i_{n+U}$ in the following way: whenever $|\lambda_{i(\alpha)}|<|\lambda_{i(\beta)}|$, then the face with index number $\alpha$ is strictly less than the face with index number $\beta$. The ordering on faces $\alpha,\beta$ such that~${|\lambda_{i(\alpha)}|=|\lambda_{i(\beta)}|}$ is chosen arbitrarily. Lemma \ref{lem:injection fi vertices to fi edges} supplies an injection $\Gamma_G:G^0_{\textnormal{fi}}\to G^1_\textnormal{int}$, where~$G^0_{\textnormal{fi}}$ is the set of fully internal vertices of $G$. 

{As argued in the paragraph before Theorem \ref{thm:no_vanishing_divdifs}, we may assume that $0\notin\{\lambda_k\}_{k=1}^\infty$. We shall apply Theorem \ref{thm:no_vanishing_divdifs} in the following four cases. }

\begin{itemize}
\item For each $v\in G^0$ which is not fully internal, Theorem \ref{thm:no_vanishing_divdifs} yields
$$|f'\{i(\alpha^v_1),\ldots,i(\alpha^v_{\deg(v)})\}|\leq c_2\max\{|\lambda_k|^{-p}~:~k\in\N_{\geq1}\}\equiv C_1,$$
for a constant $C_1$ which is independent of $i_{n+1},\ldots,i_{n+U}$ and $N$.

\item For each $v\in G^0$ which is fully internal, Theorem \ref{thm:no_vanishing_divdifs} yields
$$|f'\{i(\alpha^v_1),\ldots,i(\alpha^v_{\deg(v)})\}|\leq C_2|\lambda_{i(\alpha(v))}|^{-p},$$
where $\alpha(v)$ is the index number of a minimal face bordering $v$ (in the total order introduced above), and $C_2$ is some constant which is independent of $i_{n+1},\ldots,i_{n+U}$ and $N$.
\item For each $e\in G^1_\textnormal{int}$ which is not fully internal, Theorem \ref{thm:no_vanishing_divdifs} yields
$$|f'\{i(\beta^e_1),i(\beta^e_2)\}| 
\geq c_1\min\{|\lambda_i|^{-p}~:~i\in\{i_1,\ldots,i_n\}\}\equiv C_3,$$
for $C_3>0$ independent of $i_{n+1},\ldots,i_{n+U},N$.
\item For every edge $e$ we choose an index number $\beta(e)$ such that $\beta(\Gamma_G(v))=\alpha(v)$ for every fully internal vertex $v$, using injectivity of $\Gamma_G$. For each fully internal edge $e\in G^1_\textnormal{int}$, Theorem \ref{thm:no_vanishing_divdifs} yields
$$|f'\{i(\beta^e_1),i(\beta^e_2)\}|\geq C_4 |\lambda_{i(\beta(e))}|^{-p},$$
for $C_4>0$ independent of $i_{n+1},\ldots,i_{n+U},N$.
\end{itemize}

Our choice of $\beta(e)$ ensures that every factor in the numerator of \eqref{eq:first bound amplitude} is canceled against a factor in the denominator, and we end up with 
\begin{align*}
|\Ampl_{N,i_1,\ldots,i_n}(G)|&\lesssim\sum_{i_{n+1},\ldots,i_{n+U}=1}^N\frac{1}{\prod_{e\in G^1_{\textnormal{int}}\setminus\Gamma_G(G^0_{\textnormal{fi}})} |\lambda_{i(\beta(e))}|^{-p}}\\
&=\sum_{i_{n+1},\ldots,i_{n+U}=1}^N\prod_{e\in G^1_{\textnormal{int}}\setminus\Gamma_G(G^0_{\textnormal{fi}})} |\lambda_{i(\beta(e))}|^{p},
\end{align*}
noting that $p$ is nonnegative.
We now use that 
$$|\lambda_{i(\beta(e))}|\lesssim |\lambda_{N}|\lesssim N^{1/d}.$$
We obtain
\begin{align*}
|\Ampl_{N,i_1,\ldots,i_n}(G)|&\lesssim\sum_{i_{n+1},\ldots,i_{n+U}=1}^N\prod_{e\in G^1_{\textnormal{int}}\setminus\Gamma_G(G^0_{\textnormal{fi}})} N^{p/d}\\
&=\sum_{i_{n+1},\ldots,i_{n+U}=1}^N N^{(p/d)(E_\textnormal{fi}-V_\textnormal{fi})}\\
&=N^{U+(p/d)(E_\textnormal{fi}-V_\textnormal{fi})},
\end{align*}
as we wanted to show.
\end{proof}

\begin{rem}\label{rem:apparent shortcut}
We explain here a problem with an apparent shortcut to the above proof. Indeed, given positive~${n_1,\ldots,n_U}$ it is not hard to show (using $\lambda_j\sim j^{1/d}$) that
\begin{align}\label{eq:sums like integrals}
    \sum_{j_1=1}^N\cdots\sum_{j_U=1}^N \lambda_{j_1}^{n_1}\cdots\lambda_{j_U}^{n_U}=\O(N^{U+(n_1+\ldots+n_U)/d}),
\end{align} 
exactly as in the case with integrals instead of sums. If one skips the graph-theoretical Lemma \ref{lem:injection fi vertices to fi edges} and applies the estimates of Theorem \ref{thm:weighted_divdiff_bound} for arbitrary bordering indices of the fully internal vertices and edges, one can estimate the amplitude by the sum on the left-hand side of \eqref{eq:sums like integrals} where indeed the powers automatically add up to $n_1+\ldots+n_U=p(\Efi-\Vfi)$ as required! However, the $n_1,\ldots,n_U$ may not be positive, which invalidates \eqref{eq:sums like integrals}. The following example shows that such an estimate of the amplitude by \eqref{eq:sums like integrals} really is too coarse in general. We surely have
\begin{align}\label{eq:example diagram}
\raisebox{-20pt}{
\begin{tikzpicture}[thick]
	\draw (0.45,-.7) to (1,0);
	\draw (1.55,-.7) to (1,0);
	\draw (1,0.054) arc (-90:270:0.3cm);
	\draw (1,-0.035) arc (-90:270:.7cm);
	\node at (1,-.55) {$i_2$};
	\node at (0,1) {$i_1$};
	\vertex{1,0};
\end{tikzpicture}
}
=\lambda_{i_1}^{-1}\lambda_{i_2}^{-1}\sum_{k=1}^N\sum_{l=1}^N\frac{f'\{i_1,k,l,k,i_1,i_2\}}{f'\{i_1,k\}f'\{k,l\}}\lesssim
\sum_{k=1}^N\sum_{l=1}^N\frac{l^{-p/d}}{k^{-p/d}k^{-p/d}}=\sum_{k=1}^N k^{2p/d}\sum_{l=1}^N l^{-p/d}.
\end{align}
Naively adding up the orders gives the correct result, $\O(N^{2+\frac{p}{d}})$. But because $-p/d$ is negative, we cannot add up the powers: Assuming $p/d>1$, the sequence $(\sum_{l=1}^N l^{-p/d})_{N\in\N}$ is convergent with nonzero limit. Hence, $\sum_{l=1}^N l^{-p/d}$ is of order $N^0$, not of order $N^{1-p/d}$. The right-hand side of \eqref{eq:example diagram} is therefore $\O(N^{1+2p/d})$, and as we now know we have a better bound for the left-hand side. The trick is simply to choose the indices so that each vertex contribution is canceled by an edge contribution.
\end{rem}

\subsection{Lower bound}\label{sect: lower}
In this subsection we prove a lower bound on the divergence for general Feynman ribbon graphs {and functions whose divided differences do not vanish on the spectrum.} 
This lower bound is the last remaining part of our first main theorem, Theorem \ref{thm:main}.


\begin{thm}\label{thm:lower bound}
Assume the conditions of Theorem~\ref{thm:main}. Then, asymptotically in $N\in\N$, we have
\begin{align}\label{eqn: ampl_lower_bound}
|\Ampl_{N,i_1,\ldots,i_n}(G)|&\gtrsim N^{U+\frac{p}{d}(E_{\textnormal{fi}}-V_{\textnormal{fi}})}.
\end{align}
\end{thm}
We spend the rest of this section proving the above theorem, so we fix $f,p,\{\lambda_k\}_{k=1}^\infty,$ $c_1,c_2$, $G=(G^0,n,G^1)$, and $i_1,\ldots,i_n$, as in Theorem~\ref{thm:main}. We first reduce the proof to the case where the spectrum is located sufficiently far away from zero. Towards this, we introduce the notion of a \emph{restricted amplitude}, 
\begin{equation}\label{eqn: restricted_ampl}
    \Ampl_{N,i_1,\dotsc,i_n}^{\ge i_R}(G):=\sum_{i_{n+1},\ldots,i_{n+U}=i_R}^N\frac{\prod_{v\in G^0} f'[i(\alpha^v_1),\ldots,i(\alpha^v_{\deg(v)})]}{\prod_{e\in G^1_{\textnormal{int}}} f'[i(\beta^e_1),i(\beta^e_2)]},
\end{equation}
where the running indices are restricted to being sufficiently large. Note that the natural analogue of Lemma~\ref{lem:Ampl for weighted divdifs} holds for restricted amplitudes. 
We moreover introduce \textbf{artificially broken} graphs $G_{\mathfrak{b},\gamma}$, which are Feynman ribbon graphs $G$ together with a subset $\mathfrak{b}$ of the set of unbroken faces (called artificially broken faces), and an association $\gamma:\mathfrak{b}\to\N_{\geq1}$ of an index to each artificially broken face. By interpreting these indices as external and the elements of $\mathfrak{b}$ as broken, and by summing over the remaining unbroken indices, the definitions \eqref{eqn: restricted_ampl} and \eqref{eq:omg} naturally extend to artificially broken graphs.
\begin{lem}\label{lem: restricted_lower_sufficient}
    Assume the conditions of Theorem~\ref{thm:main} and fix a Feynman ribbon graph $G$. Then there exists some $i_R\in \N$ such that all terms of \eqref{eqn: restricted_ampl} have the same sign. Moreover, if there exist~$M,c_3'>0$ such that for all~${N\ge M}$ we have \begin{equation}\label{eqn: assume_lower_restr}|\Ampl_{N,i_1,\dotsc,i_k,\gamma}^{\ge i_R}(G_{\mathfrak{b}})|\ge c_3'N^{\omega(G_{\mathfrak{b}})}
    \end{equation} 
    for all artificially broken Feynman ribbon graphs $G_{\mathfrak{b},\gamma}$ obtained from $G$ with 
    external indices $i_1,\dotsc,i_k< i_R$, 
     then there also exists some $c_3>0$ such that \begin{equation*}|\Ampl_{N,i_1,\dotsc,i_n}(G)|\ge c_3N^{\omega(G)}
    \end{equation*} 
    for all $N\ge M$.
\end{lem}
\begin{proof}
%
Let $m:=\max_{v\in G^0}\deg(v)$. Using Lemma~\ref{lem:limit}, we may fix some $R_I:=R_{i_1,\dotsc,i_n}$ such that for all~$R_I<|\lambda_{i_{n+1}}|\le\dotsc\le |\lambda_{n+U}|$ we have \begin{equation*}
    \sgn f'\{\lambda_{i_{j_1}},\dotsc \lambda_{i_{j_l}}, \lambda_{i_{k_1}},\dotsc\lambda_{i_{k_{l'}}}\}=\sgn f'\{\lambda_{i_{j_1}},\dotsc \lambda_{i_{j_l}}\}
\end{equation*} 
for any choice of $\{j_1,\dotsc,j_l\}\subseteq\{1,\dotsc,n\}$, $\{k_1,\dotsc,k_{l'}\}\subseteq\{n+1,\dotsc,n+U\}$, $1\le l'\le m-l$.
Fix some~$i_R\in\N$ such that $|\lambda_i|>R_I$ for all $i\ge i_R$. We can split up the sum of a graph amplitude as 
\begin{align}
    \sum_{i_{n+1},\ldots,i_{n+U}=1}^N &= \sum_{i_{n+1},\ldots,i_{n+U}=i_R}^N\nonumber \\
    &+ \sum_{i_{n+1}=1}^{i_{R}-1} \sum_{i_{n+2},\ldots,i_{n+U}=i_R}^N + \sum_{i_{n+2}=1}^{i_{R}-1} \sum_{i_{n+1}, i_{n+3},\ldots,i_{n+U}=i_R}^N +\cdots\nonumber\\
    &+\sum_{i_{n+1}, i_{n+2}=1}^{i_{R}-1} \sum_{i_{n+3},\ldots,i_{n+U}=i_R}^N+\cdots\label{eqn: sum_split}\\
    &+\cdots\nonumber\\
    &+ \sum_{i_{n+1},\ldots,i_{n+U}=1}^{i_R-1}.\nonumber
\end{align}
A succinct way to write this is
\begin{align}
    \sum_{i_{n+1},\ldots,i_{n+U}=1}^N &=\sum_{\bfr\subseteq\mathcal U}\sum_{\substack {i_b=1\\(b\in\bfr)}}^{i_R-1}\sum_{\substack{i_u=i_R\\(u\in\mathcal U\setminus\bfr)}}^N,
\end{align}
where 
$(b\in\bfr)$ and $(u\in \mathcal{U}\setminus\bfr)$ indicate the repeated sums, and $\mathcal U$ is the set of unbroken faces of $G$. By slight abuse of notation, we will use~$b,u$ to refer to either faces in $\mathcal{U}$ or their associated index numbers.
Note that each of the internal (repeated) sums
\begin{align}\label{eq:artificially broken}
\sum_{\substack{i_u=i_R\\(u\in\mathcal U\setminus\bfr)}}^N\frac{\prod_{v\in G^0} f'\{i(\alpha^v_1),\ldots,i(\alpha^v_{\deg(v)})\}}{\prod_{e\in G^1_{\textnormal{int}}} f'\{i(\beta^e_1),i(\beta^e_2)\}}
\end{align}
describes the restricted amplitude of an artificially broken graph $G_{\bfr,\gamma}$, 
for which the associated fixed indices are given by ${\gamma: \bfr\to \{1,\dotsc,i_R-1\}}$, $\gamma(b):=i_{b}$. Note that by construction of $i_R$, all summands in~\eqref{eq:artificially broken} have the same sign $\sigma_{b,\gamma}\in\{\pm1\}$. Hence we obtain 
\begin{equation*}
    \Ampl_{N,i_1,\ldots,i_n}(G) 
    =\sum_{\bfr\subseteq \mathcal{U}}\sum_{\gamma:\bfr\to\{1,\dotsc,i_R-1\}}\sigma_{b,\gamma}|\Ampl_{N,{i_1},\dotsc,i_n,\gamma}^{\ge i_R}(G_{\bfr})|.
\end{equation*}
Now assume that there exist~$M,c_3'>0$ such that for all~${N\ge M}$ we have~\eqref{eqn: assume_lower_restr}. Combined with Theorem~\ref{thm:upper}, this yields a two-sided bound on all restricted amplitudes in the double sum above, 
and hence there exist constants $c_{b,\gamma}>0$ such that
\begin{equation*}
    |\Ampl_{N,i_1,\ldots,i_n}(G)| \ge \left|\sum_{\bfr\subseteq \mathcal{U}}\sum_{\gamma:\bfr\to\{1,\dotsc,i_R-1\}}\sigma_{b,\gamma}c_{b,\gamma}N^{\omega(G_{\bfr})}\right|.
\end{equation*}

Note that if $\bfr\neq\emptyset$, we have $\omega(G_{\bfr})\le U-1+(p/d)(E_{\f}-V_{\f})=\omega(G)-1$, hence the asympotic behaviour of the amplitude of $G$ is dominated by the restricted amplitude corresponding to~$G_{\bfr}$ with $\bfr=\emptyset$, which concludes the proof.
\end{proof}
We use the remainder of this section to show that~\eqref{eqn: assume_lower_restr} holds for an artificially broken graph~$G_{\mathfrak{b},\gamma}$. However, as the choices of $\mathfrak{b}$ and $\gamma$ play no further role, we suppress them in the notation. Likewise, the external indices of the artificially broken graph (given by $i_1,\ldots,i_k$ and $\gamma(b)=i_b$ ($b\in\mathfrak{b}$)) are denoted $i_1,\ldots,i_n$, and the amount of remaining unbroken faces by $U$.

We employ the first conclusion of Lemma \ref{lem: restricted_lower_sufficient}. As all terms in the sum 
\begin{equation*}
    |\Ampl^{\ge i_R}_{N,i_1,\dotsc,i_n}(G)|=|\lambda_{i_1}^{-1}\cdots\lambda_{i_n}^{-1}|\sum_{i_{n+1},\ldots,i_{n+U}=i_R}^N\frac{\prod_{v\in G^0} |f'\{i(\alpha^v_1),\ldots,i(\alpha^v_{\deg(v)})\}|}{\prod_{e\in G^1_{\textnormal{int}}} |f'\{i(\beta^e_1),i(\beta^e_2)\}|} 
\end{equation*}
are positive, we can further restrict our running indices to any subset. In particular, for $N$ large enough such that $\lfloor N/2\rfloor\ge i_R$ we have \begin{align}\label{eq:Ampl}
|\Ampl^{\ge i_R}_{N,i_1,\ldots,i_n}(G)|\gtrsim\sum_{\substack{i_{n+1},\ldots,i_{n+U}\in\{\lfloor N/2\rfloor,\ldots,N\}\\ |\lambda_{i_{n+1}}|\leq\cdots\leq|\lambda_{i_{n+U}}|}}\frac{\prod_{v\in G^0} |f'\{i(\alpha^v_1),\ldots,i(\alpha^v_{\deg(v)})\}|}{\prod_{e\in G^1_{\textnormal{int}}}|f'\{i(\beta^e_1),i(\beta^e_2)\}|}.
\end{align}

\begin{lem}\label{lem:amount of terms}
    There exist $M,c_1>0$ such that, for all $N\geq M$,
    $$\sum_{\substack{i_{n+1},\ldots,i_{n+U}\in\{\lfloor N/2\rfloor,\ldots,N\}\\ |\lambda_{i_{n+1}}|\leq\cdots\leq|\lambda_{i_{n+U}}|}}1\geq c_1N^U.$$
    Moreover, there exist $M,c_2>0$ such that for each $N\geq M$ and for each term in the sum, we have $$|\lambda_{i_{n+1}}|\geq c_2 N^{1/d}.$$
\end{lem}
\begin{proof}
First, there are at least $\vect{\lfloor N/2\rfloor}{U}$
 options to choose $i_{n+1},\ldots,i_{n+U}$ out of $\{\lfloor N/2\rfloor,\ldots,N\}$ such that $|\lambda_{i_{n+1}}|\leq\cdots\leq|\lambda_{i_{n+U}}|$. Clearly,
${\lfloor N/2\rfloor \choose U}\gtrsim{N/4 \choose U}\gtrsim \frac{(N/4)!}{(N/4-U)!}\geq (N/4-U)^U\gtrsim N^U.$
Second, we notice that for $i\geq \lfloor N/2\rfloor $ we have $|\lambda_i|\gtrsim \lfloor N/2\rfloor^{1/d}\gtrsim N^{1/d}$ by assumption on the growth of the eigenvalues.
%
\end{proof}

We come to the core of the proof of Theorem \ref{thm:lower bound}, the lower bound part of our first power counting theorem. It remains to combine \eqref{eq:Ampl} with the graph-theoretical Lemma \ref{lem:injection fi vertices to fi edges}, the above lemmas, and the results of Section \ref{sct:estimates}.
\begin{proof}[Proof of Theorem \ref{thm:lower bound}]
By Theorem~\ref{thm:no_vanishing_divdifs}, if both indices $\beta_1^e$ and $\beta_2^e$ are running indices, we have 
$$|f'\{i(\beta_1^e),i(\beta_2^e)\}|\lesssim |\lambda_{\min(i(\beta_1^e), i(\beta_2^e))}|^{-p}.$$
If at least one of the indices is fixed, we have the bound $f'\{i(\beta_1^e),i(\beta_2^e)\}\lesssim 1$. Thus in the denominator of~\eqref{eq:Ampl} we obtain $E_{\textnormal{fi}}$ factors of $|\lambda_{i_{n+1}}|^{-p}$.

Furthermore, by Theorem~\ref{thm:no_vanishing_divdifs} we know that
\begin{align}\label{eq:vertex bound}
|f'\{i(\alpha^v_1),\ldots,i(\alpha^v_{\deg(v)})\}|\gtrsim |\lambda_{i}|^{-p}\quad \textnormal{for all }i\in\{i(\alpha^v_1),\ldots,i(\alpha^v_{\deg(v)})\}.
\end{align}
If one of the indices is fixed we thus have the bound $|{f'\{i(\alpha^v_1),\ldots,i(\alpha^v_{\deg(v)})\}|\gtrsim 1}$. In the numerator of \eqref{eq:Ampl} we therefore obtain $V_{\textnormal{fi}}$ factors of $|\lambda_i|^{-p}$, for possibly different $i$'s. 

We employ the injection $\Gamma_G$ from Lemma \ref{lem:injection fi vertices to fi edges} from fully internal vertices to fully internal edges. For any fully internal vertex $v$ the edge $e=\Gamma_G(v)$ yields $|\lambda_{\min(i(\beta_1^e), i(\beta_2^e))}|^{-p}$ in the denominator. 
For the vertex-bound \eqref{eq:vertex bound} above we can choose the index arbitrarily, hence we also obtain the factor~$|\lambda_i|^{-p}$ in the numerator for the same $i$ (where $i=\min(i(\beta_1^e), i(\beta_2^e))$). 
This happens for any fully internal vertex~$v$, so all $V_{\textnormal{fi}}$ factors of $|\lambda_i|^{-p}$ in the numerator cancel against~$V_{\textnormal{fi}}$ of the $E_{\textnormal{fi}}$ factors of $|\lambda_i|^{-p}$ in the denominator. We are left with only $E_{\textnormal{fi}}-V_{\textnormal{fi}}$ factors of $|\lambda_i|^{-p}$ in the denominator, i.e., $E_{\textnormal{fi}}-V_{\textnormal{fi}}$ factors of $|\lambda_i|^{p}$. Note that $E_{\textnormal{fi}}-V_{\textnormal{fi}}\geq0$, and that $|\lambda_i|^{p}\geq |\lambda_{i_{n+1}}|^{p}$. From~\eqref{eq:Ampl} we thus obtain
\begin{align*}
|\Ampl^{\ge i_R}_{N,i_1,\ldots,i_n}(G)|&\gtrsim\sum_{\substack{i_{n+1},\ldots,i_{n+U}\in\{\lfloor N/2\rfloor,\ldots,N\}\\
|\lambda_{i_{n+1}}|\leq\cdots\leq|\lambda_{i_{n+U}}|}}|\lambda_{i_{n+1}}|^{(E_{\textnormal{fi}}-V_{\textnormal{fi}})p}.
\end{align*}
By Lemma \ref{lem:amount of terms},  we have $|\lambda_{i_{n+1}}|\gtrsim N^{1/d}$ for all $i_{n+1}$ in the above sum. Since $p(E_{\textnormal{fi}}-V_{\textnormal{fi}})\geq0$, we obtain
\begin{align*}
|\Ampl^{\ge i_R}_{N,i_1,\ldots,i_n}(G)|
&\gtrsim \sum_{\substack{i_{n+1},\ldots,i_{n+U}\in\{\lfloor N/2\rfloor,\ldots,N\}\\
|\lambda_{i_{n+1}}|\leq\cdots\leq|\lambda_{i_{n+U}}|}}N^{(E_{\textnormal{fi}}-V_{\textnormal{fi}})p/d}\\
&=N^{(E_{\textnormal{fi}}-V_{\textnormal{fi}})p/d}\sum_{\substack{i_{n+1},\ldots,i_{n+U}\in\{\lfloor N/2\rfloor,\ldots,N\}\\
|\lambda_{i_{n+1}}|\leq\cdots\leq|\lambda_{i_{n+U}}|}}1.
\end{align*}
An application of Lemma \ref{lem:amount of terms} to the above estimate yields
\begin{align*}
|\Ampl^{\ge i_R}_{N,i_1,\ldots,i_n}(G)|
&\gtrsim N^{U+\frac{p}{d}(E_{\textnormal{fi}}-V_{\textnormal{fi}})},
\end{align*}
which by Lemma~\ref{lem: restricted_lower_sufficient} is sufficient to conclude the proof.
\end{proof}

\begin{rem}\label{rem:2-point 2-loop}
From Theorem \ref{thm:main} we may derive, at each loop order, the set of relevant diagrams. For example, the set of 2-point 2-loop diagrams (with vertices of valence $\geq3$) with maximal order of divergence is
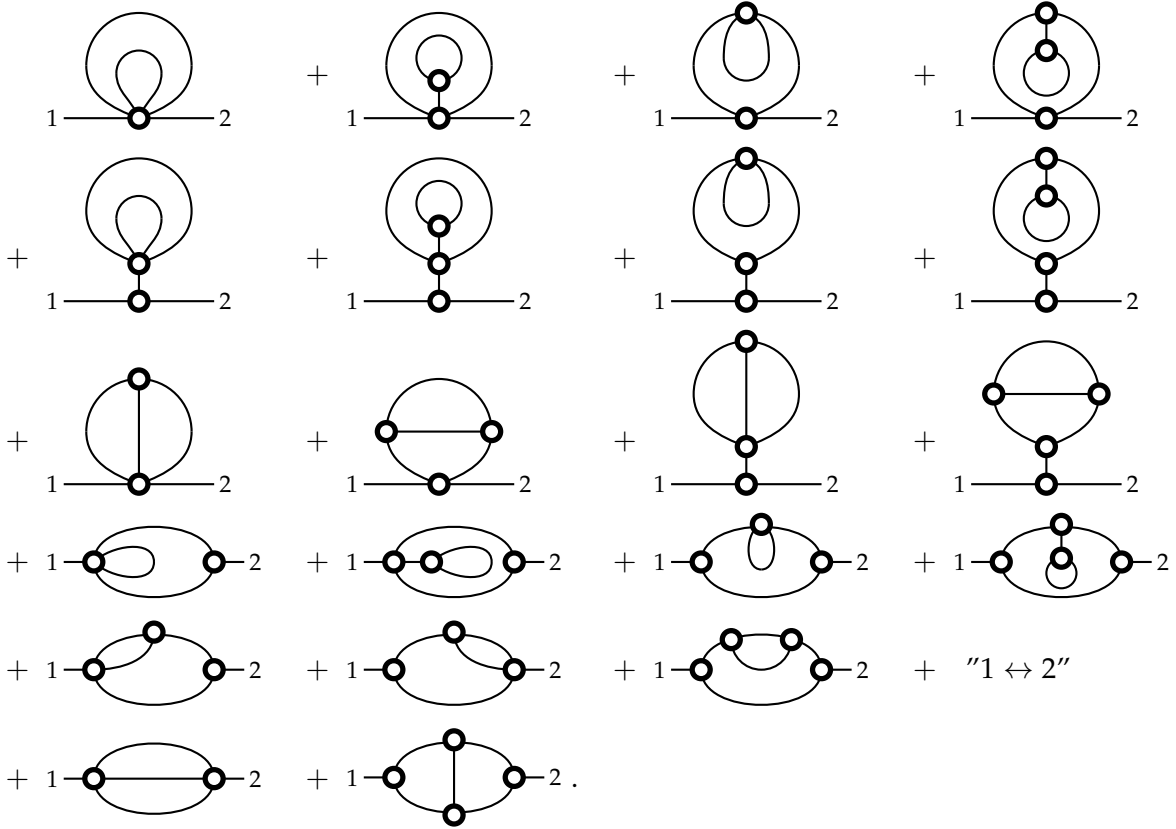
\begin{align*}
~&
\raisebox{-20pt}{
\begin{tikzpicture}[thick]
\draw (-1,0) to (1,0);
\draw (-.7,.7) arc (180:0:.7cm);
\draw (0,0) to[out=160,in=-90] (-.7,.7);
\draw (0,0) to[out=20,in=-90] (.7,.7);
\draw (-.3,.6) arc (180:0:.3cm);
\draw (0,0) to[out=110,in=-90] (-.3,.6);
\draw (0,0) to[out=70,in=-90] (.3,.6);
\vertex{0,0};
\node at (-1.15,0) {\footnotesize $1$};
\node at (1.15,0) {\footnotesize$2$};
\end{tikzpicture}}
&+&
\raisebox{-20pt}{
\begin{tikzpicture}[thick]
\draw (-1,0) to (1,0);
\draw (-.7,.7) arc (180:0:.7cm);
\draw (0,0) to[out=160,in=-90] (-.7,.7);
\draw (0,0) to[out=20,in=-90] (.7,.7);
\draw (0,0) to (0,.5);
\draw (0,.5) arc (-90:270:.3cm);
\vertex{0,0};
\vertex{0,.5};
\node at (-1.15,0) {\footnotesize $1$};
\node at (1.15,0) {\footnotesize$2$};
\end{tikzpicture}}
&+&
\raisebox{-20pt}{
\begin{tikzpicture}[thick]
\draw (-1,0) to (1,0);
\draw (-.7,.7) arc (180:0:.7cm);
\draw (0,0) to[out=160,in=-90] (-.7,.7);
\draw (0,0) to[out=20,in=-90] (.7,.7);
\draw (-.3,.8) arc (180:360:.3cm);
\draw (0,1.4) to[out=-160,in=90] (-.3,.8);
\draw (0,1.4) to[out=-20,in=90] (.3,.8);
\vertex{0,0};
\vertex{0,1.4};
\node at (-1.15,0) {\footnotesize $1$};
\node at (1.15,0) {\footnotesize$2$};
\end{tikzpicture}}
&+&
\raisebox{-20pt}{
\begin{tikzpicture}[thick]
\draw (-1,0) to (1,0);
\draw (-.7,.7) arc (180:0:.7cm);
\draw (0,0) to[out=160,in=-90] (-.7,.7);
\draw (0,0) to[out=20,in=-90] (.7,.7);
\draw (0,1.4) to (0,.9);
\draw (0,.9) arc (90:450:.3cm);
\vertex{0,0};
\vertex{0,1.4};
\vertex{0,.9};
\node at (-1.15,0) {\footnotesize $1$};
\node at (1.15,0) {\footnotesize$2$};
\end{tikzpicture}}
\\
+&
\raisebox{-20pt}{
\begin{tikzpicture}[thick]
\draw (-1,-.5) to (1,-.5);
\draw (0,-.5) to (0,0);
\draw (-.7,.7) arc (180:0:.7cm);
\draw (0,0) to[out=160,in=-90] (-.7,.7);
\draw (0,0) to[out=20,in=-90] (.7,.7);
\draw (-.3,.6) arc (180:0:.3cm);
\draw (0,0) to[out=110,in=-90] (-.3,.6);
\draw (0,0) to[out=70,in=-90] (.3,.6);
\vertex{0,0};
\vertex{0,-.5};
\node at (-1.15,-.5) {\footnotesize $1$};
\node at (1.15,-.5) {\footnotesize$2$};
\end{tikzpicture}}
&+&
\raisebox{-20pt}{
\begin{tikzpicture}[thick]
\draw (-1,-.5) to (1,-.5);
\draw (0,-.5) to (0,0);
\draw (-.7,.7) arc (180:0:.7cm);
\draw (0,0) to[out=160,in=-90] (-.7,.7);
\draw (0,0) to[out=20,in=-90] (.7,.7);
\draw (0,0) to (0,.5);
\draw (0,.5) arc (-90:270:.3cm);
\vertex{0,0};
\vertex{0,.5};
\vertex{0,-.5};
\node at (-1.15,-.5) {\footnotesize $1$};
\node at (1.15,-.5) {\footnotesize$2$};
\end{tikzpicture}}
&+&
\raisebox{-20pt}{
\begin{tikzpicture}[thick]
\draw (-1,-.5) to (1,-.5);
\draw (0,-.5) to (0,0);
\draw (-.7,.7) arc (180:0:.7cm);
\draw (0,0) to[out=160,in=-90] (-.7,.7);
\draw (0,0) to[out=20,in=-90] (.7,.7);
\draw (-.3,.8) arc (180:360:.3cm);
\draw (0,1.4) to[out=-160,in=90] (-.3,.8);
\draw (0,1.4) to[out=-20,in=90] (.3,.8);
\vertex{0,0};
\vertex{0,1.4};
\vertex{0,-.5};
\node at (-1.15,-.5) {\footnotesize $1$};
\node at (1.15,-.5) {\footnotesize$2$};
\end{tikzpicture}}
&+&
\raisebox{-20pt}{
\begin{tikzpicture}[thick]
\draw (-1,-.5) to (1,-.5);
\draw (0,-.5) to (0,0);
\draw (-.7,.7) arc (180:0:.7cm);
\draw (0,0) to[out=160,in=-90] (-.7,.7);
\draw (0,0) to[out=20,in=-90] (.7,.7);
\draw (0,1.4) to (0,.9);
\draw (0,.9) arc (90:450:.3cm);
\vertex{0,0};
\vertex{0,1.4};
\vertex{0,.9};
\vertex{0,-.5};
\node at (-1.15,-.5) {\footnotesize $1$};
\node at (1.15,-.5) {\footnotesize$2$};
\end{tikzpicture}}
\\
+&
\raisebox{-20pt}{
\begin{tikzpicture}[thick]
\draw (-1,0) to (1,0);
\draw (-.7,.7) arc (180:0:.7cm);
\draw (0,0) to[out=160,in=-90] (-.7,.7);
\draw (0,0) to[out=20,in=-90] (.7,.7);
\draw (0,1.4) to (0,0);
\vertex{0,0};
\vertex{0,1.4};
\node at (-1.15,0) {\footnotesize $1$};
\node at (1.15,0) {\footnotesize$2$};
\end{tikzpicture}}
&+&
\raisebox{-20pt}{
\begin{tikzpicture}[thick]
\draw (-1,0) to (1,0);
\draw (-.7,.7) arc (180:0:.7cm);
\draw (0,0) to[out=160,in=-90] (-.7,.7);
\draw (0,0) to[out=20,in=-90] (.7,.7);
\draw (-.7,.7) to (.7,.7);
\vertex{0,0};
\vertex{-.7,.7};
\vertex{.7,.7};
\node at (-1.15,0) {\footnotesize $1$};
\node at (1.15,0) {\footnotesize$2$};
\end{tikzpicture}}
&+&
\raisebox{-20pt}{
\begin{tikzpicture}[thick]
\draw (-1,-.5) to (1,-.5);
\draw (0,-.5) to (0,0);
\draw (-.7,.7) arc (180:0:.7cm);
\draw (0,0) to[out=160,in=-90] (-.7,.7);
\draw (0,0) to[out=20,in=-90] (.7,.7);
\draw (0,1.4) to (0,0);
\vertex{0,0};
\vertex{0,1.4};
\vertex{0,-.5};
\node at (-1.15,-.5) {\footnotesize $1$};
\node at (1.15,-.5) {\footnotesize$2$};
\end{tikzpicture}}
&+&
\raisebox{-20pt}{
\begin{tikzpicture}[thick]
\draw (-1,-.5) to (1,-.5);
\draw (0,-.5) to (0,0);
\draw (-.7,.7) arc (180:0:.7cm);
\draw (0,0) to[out=160,in=-90] (-.7,.7);
\draw (0,0) to[out=20,in=-90] (.7,.7);
\draw (-.7,.7) to (.7,.7);
\vertex{0,0};
\vertex{0,-.5};
\vertex{-.7,.7};
\vertex{.7,.7};
\node at (-1.15,-.5) {\footnotesize $1$};
\node at (1.15,-.5) {\footnotesize$2$};
\end{tikzpicture}}
\\
+&
\raisebox{-15pt}{
\begin{tikzpicture}[thick]
\draw (-1.2,0) to (-.8,0);
\draw (-.8,0) to[out=90,in=90] (.8,0);
\draw (-.8,0) to[out=-90,in=-90] (.8,0);
\draw (-.8,0) to[out=45,in=90] (0,0);
\draw (-.8,0) to[out=-45,in=-90] (0,0);
\draw (.8,0) to (1.2,0);
\vertex{-.8,0};
\vertex{.8,0};
\node at (-1.35,0) {\footnotesize $1$};
\node at (1.35,0) {\footnotesize$2$};
\end{tikzpicture}}
&+&
\raisebox{-15pt}{
\begin{tikzpicture}[thick]
\draw (-1.2,0) to (-.8,0);
\draw (-.8,0) to[out=90,in=90] (.8,0);
\draw (-.8,0) to[out=-90,in=-90] (.8,0);
\draw (-.8,0) to (-0.3,0);
\draw (-.3,0) to[out=45,in=90] (.5,0);
\draw (-.3,0) to[out=-45,in=-90] (.5,0);
\draw (.8,0) to (1.2,0);
\vertex{-.8,0};
\vertex{.8,0};
\vertex{-.3,0};
\node at (-1.35,0) {\footnotesize$1$};
\node at (1.35,0) {\footnotesize$2$};
\end{tikzpicture}}
&+&
\raisebox{-15pt}{
\begin{tikzpicture}[thick]
\draw (-1.2,0) to (-.8,0);
\draw (-.8,0) to[out=90,in=90] (.8,0);
\draw (-.8,0) to[out=-90,in=-90] (.8,0);
\draw (0,.5) to[out=-160,in=180] (0,-.1);
\draw (0,.5) to[out=-20,in=0] (0,-.1);
\draw (.8,0) to (1.2,0);
\vertex{-.8,0};
\vertex{.8,0};
\vertex{0,.5};
\node at (-1.35,0) {\footnotesize$1$};
\node at (1.35,0) {\footnotesize$2$};
\end{tikzpicture}}
&+&
\raisebox{-15pt}{
\begin{tikzpicture}[thick]
\draw (-1.2,0) to (-.8,0);
\draw (-.8,0) to[out=90,in=90] (.8,0);
\draw (-.8,0) to[out=-90,in=-90] (.8,0);
\draw (0,.5) to (0,.05);
\draw (0,.05) arc (90:450:.2cm);
\draw (.8,0) to (1.2,0);
\vertex{-.8,0};
\vertex{.8,0};
\vertex{0,.5};
\vertex{0,.05};
\node at (-1.35,0) {\footnotesize$1$};
\node at (1.35,0) {\footnotesize$2$};
\end{tikzpicture}}
\\
+&
\raisebox{-15pt}{
\begin{tikzpicture}[thick]
\draw (-1.2,0) to (-.8,0);
\draw (-.8,0) to[out=90,in=90] (.8,0);
\draw (-.8,0) to[out=-90,in=-90] (.8,0);
\draw (-.8,0) to[out=0,in=-90] (0,.5);
\draw (.8,0) to (1.2,0);
\vertex{0,.5};
\vertex{-.8,0};
\vertex{.8,0};
\node at (-1.35,0) {\footnotesize $1$};
\node at (1.35,0) {\footnotesize$2$};
\end{tikzpicture}}
&+&
\raisebox{-15pt}{
\begin{tikzpicture}[thick]
\draw (-1.2,0) to (-.8,0);
\draw (-.8,0) to[out=90,in=90] (.8,0);
\draw (-.8,0) to[out=-90,in=-90] (.8,0);
\draw (.8,0) to[out=180,in=-90] (0,.5);
\draw (.8,0) to (1.2,0);
\vertex{0,.5};
\vertex{-.8,0};
\vertex{.8,0};
\node at (-1.35,0) {\footnotesize $1$};
\node at (1.35,0) {\footnotesize$2$};
\end{tikzpicture}}
&+&
\raisebox{-15pt}{
\begin{tikzpicture}[thick]
\draw (-1.2,0) to (-.8,0);
\draw (-.8,0) to[out=90,in=90] (.8,0);
\draw (-.8,0) to[out=-90,in=-90] (.8,0);
\draw (.8,0) to (1.2,0);
\draw (-.4,.4) arc (180:360:.4cm);
\vertex{-.4,.4};
\vertex{.4,.4};
\vertex{-.8,0};
\vertex{.8,0};
\node at (-1.35,0) {\footnotesize $1$};
\node at (1.35,0) {\footnotesize$2$};
\end{tikzpicture}}
&+&
\quad"1\leftrightarrow2"
\\
+&
\raisebox{-15pt}{
\begin{tikzpicture}[thick]
\draw (-1.2,0) to (-.8,0);
\draw (-.8,0) to[out=90,in=90] (.8,0);
\draw (-.8,0) to[out=-90,in=-90] (.8,0);
\draw (-.8,0) to (.8,0);
\draw (.8,0) to (1.2,0);
\vertex{-.8,0};
\vertex{.8,0};
\node at (-1.35,0) {\footnotesize $1$};
\node at (1.35,0) {\footnotesize$2$};
\end{tikzpicture}}
&+&
\raisebox{-15pt}{
\begin{tikzpicture}[thick]
\draw (-1.2,0) to (-.8,0);
\draw (-.8,0) to[out=90,in=90] (.8,0);
\draw (-.8,0) to[out=-90,in=-90] (.8,0);
\draw (.8,0) to (1.2,0);
\draw (0,.5) to (0,-.5);
\vertex{0,.5};
\vertex{0,-.5};
\vertex{-.8,0};
\vertex{.8,0};
\node at (-1.35,0) {\footnotesize $1$};
\node at (1.35,0) {\footnotesize$2$};
\end{tikzpicture}}.
\end{align*}
Each has an order of divergence of $2+\frac{p}{d}$ (and no less).
One notices these diagrams are automatically planar. They are moreover connected and, though not necessarily 1PI, satisfy a similar connectivity property. Namely, their dual graphs stay connected after removing the vertices that correspond to the broken faces of the original graph. For instance, the diagram
$$
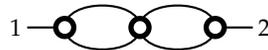
\begin{tikzpicture}[thick]
\draw (-1.5,0) to (-1,0);
\draw (-1,0) to[out=90,in=90] (0,0);
\draw (-1,0) to[out=-90,in=-90] (0,0);
\draw (1,0) to[out=90,in=90] (0,0);
\draw (1,0) to[out=-90,in=-90] (0,0);
\draw (1,0) to (1.5,0);
\vertex{-1,0};
\vertex{0,0};
\vertex{1,0};
\node at (-1.65,0) {\footnotesize $1$};
\node at (1.65,0) {\footnotesize$2$};
\end{tikzpicture}$$
does not appear, because its two unbroken faces do not share an edge.
\end{rem}

\section{Power counting when allowing eigenvalues with vanishing derivative}\label{sct:main2}
In our first main theorem, Theorem \ref{thm:main}, we assumed that $f'[\lambda_{k_1},\ldots,\lambda_{k_n}]\neq0$. With $\{\lambda_k\}_{k\geq1}$ in general position, this assumption is typically satisfied. However, there will be at least one~${x\in\R}$ such that $f'(x)=0$, which means that the assumption of nonvanishing divided differences is sensitive to infinitesimal changes in $\{\lambda_k\}_{k\geq1}$. In fact, we shall show that the actual order of divergence of many graphs is sensitive to the question of whether or not $f'$ vanishes on $\{\lambda_k\}_{k\geq1}$.

We will refer to faces with an index $i_0$ such that $f'(\lambda_{i_0})=0$ as \emph{0-faces} or \emph{singular faces}, and similarly for indices and eigenvalues.

As in Section~\ref{sect: lower}, given a subset $\bfr\subseteq\mathcal U$ of the set of unbroken faces of a Feynman ribbon graph~$G$, we let~$G_{\bfr}$ be the graph obtained from $G$ by artificially declaring the faces in $\mathfrak{b}$ to be broken. 
Elements of $\mathfrak{b}$ are called 0-faces.
\begin{thm}\label{thm:main2}
Let $f\in C^{\infty}(\mathbb{R})_\R$ be an even function with $f'$ of precise order $-p-1$ for some $p \in\mathbb{R}_{\ge0}$.
Let~$\{\lambda_k\}_{k=1}^\infty$ be a sequence of real numbers with constants $K,c_1,c_2>0$ such that $c_1k^{1/d}\leq |\lambda_k|\leq c_2k^{1/d}$ for all $k\geq K$. Assume that for all $k,l\in\N_{\geq1}$ we have $f'[\lambda_{k},\lambda_{l}]\neq0$. Then, for any Feynman ribbon graph $G=(G^0,n,G^1)$ whose vertices have valence $\geq3$, for all external indices $i_1,\ldots,i_n\in\N_{\geq1}$, there exist $M,c_4>0$ such that for all~$N\geq M$,
$$|\Ampl_{N,i_1,\ldots,i_n}(G)|\leq c_4 N^{\tilde\omega(G)},$$
$$\tilde \omega(G):=\max_{\bfr\subseteq\mathcal U}\omega_{\bfr}(G_{\bfr}):=\max_{\bfr\subseteq\mathcal U}(U^{\bfr}+\frac{p}{d}(E^{\bfr}_{\textnormal{fi}}-V^{\bfr}_{\textnormal{fi}})+\frac{p+1}{d}(E^{\bfr}_{10}-V^{\bfr}_{10})),$$
where $U^{\bfr}$ is the number of unbroken faces of $G_{\bfr}$, $E^{\bfr}_{\textnormal{fi}}$ is the number of fully internal edges of $G_{\bfr}$ (propagators bordered on both sides by unbroken faces) and $V^{\bfr}_{\textnormal{fi}}$ is the number of fully internal vertices of $G_{\bfr}$ (vertices bordered on all sides by unbroken faces). Respectively, $E^{\bfr}_{10}$ and $V^{\bfr}_{10}$ are the number of edges and vertices of~$G_{\bfr}$ that border exactly one 0-face (i.e.\ a face with index $\lambda_{i_0}$ such that $f'(\lambda_{i_0})=0$) and for the rest unbroken faces.
\end{thm}
\begin{proof}
Using the precise order of $f'$, there exists $R>0$ such that, for all $x\in\mathbb{R}\setminus[-R,R]$ and all~$0\le k\le \max_{v\in G^{0}}\;\deg(v)$, equation \eqref{eqn: precise_order} is satisfied and the conclusion of Lemma \ref{lem:zero-index bounds} holds. Let~$i_R$ be the lowest index such that $|\lambda_i| > R$ for $i\ge i_R$. For brevity we assume -- without loss of generality -- that $f'(\lambda_0)=0$ and $f'(\lambda_k)\neq0$ for~${k\geq 1}$. As in the proof of Lemma~\ref{lem: restricted_lower_sufficient}, a sum splitting argument yields
\begin{equation*}
    \Ampl_{N,i_1,\ldots,i_n}(G) = \sum_{\bfr\subseteq \mathcal{U}}\sum_{\gamma:\bfr\to\{0,\dotsc,i_R-1\}}\Ampl^{\ge i_R}_{N,{i_1},\dotsc,i_n,\gamma}(G_{\bfr}).
\end{equation*}
We claim that it suffices to show that \begin{align}\label{eqn: ampl_Gbgamma}
|\Ampl^{\ge i_R}_{N,{i_1},\dotsc,i_n,\gamma}(G_{\bfr})|&\lesssim N^{\omega_{\bfr,\gamma}(G_{\bfr})},\\
\omega_{\bfr,\gamma}(G_{\bfr})&:=(U^{\bfr}+\frac{p}{d}(E_{\textnormal{fi}}^b-V_{\textnormal{fi}}^b)+\frac{p+1}{d}(E_{10}^{\bfr,\gamma}-V_{10}^{\bfr,\gamma})),\nonumber
\end{align}
as the following graph-theoretical argument shows that $\omega_{\bfr,\gamma}(G_{\bfr})\le \omega_{\bfr}(G_{\bfr})$ for all such $\gamma$. In particular, if we let $\gamma_0$ denote the map that sends all artificially broken faces to the zero index, then $\omega_{\bfr,\gamma_0}(G_{\bfr})= \omega_{\bfr}(G_{\bfr})$, hence this bound is sharp.

More generally, we have $\omega_{\bfr}(G_{\bfr,\gamma})\leq\omega_{\bfr}(G_{\bfr,\gamma'})$ whenever $\gamma'$ is obtained from $\gamma$ by setting \begin{equation*}
    \gamma'(b)=\begin{cases}
        0,&b=b', \\
        \gamma(b),&\text{else},
    \end{cases}
\end{equation*}
for some $b'$ with $\gamma(b')\neq0$. Indeed, as the values $U^{\bfr}$, $\Efi^b$, and $\Vfi^b$ do not depend on $\gamma$, it suffices to show \begin{equation*}
E_{10}^{\bfr,\gamma}-V_{10}^{\bfr,\gamma}\leq E_{10}^{\bfr,\gamma'}-V_{10}^{\bfr,\gamma'}.
\end{equation*}
Towards this, first note that $E_{10}^{\bfr,\gamma'}-E_{10}^{\bfr,\gamma}$ equals the number of edges that border the face $b'$ on the one side and border an unbroken face (i.e.\ an element of $\mathcal{U}\setminus\bfr$) on the other. Moreover, $V_{10}^{\bfr,\gamma'}-V_{10}^{\bfr,\gamma}$ equals the number of vertices that border $b'$ and for the rest only unbroken faces. 
For each distinct vertex of the latter kind there is at least one distinct edge of the former kind, implying that $$ V_{10}^{\bfr,\gamma'}-V_{10}^{\bfr,\gamma}\leq E_{10}^{\bfr,\gamma'}-E_{10}^{\bfr,\gamma},$$which concludes the argument.

It remains to prove~\eqref{eqn: ampl_Gbgamma}. The strategy is similar to the proof of Theorem \ref{thm:upper}, but with different bounds for the edges and vertices, in which we distinguish three different types of indices and their corresponding eigenvalues.
We distinguish between (i) eigenvalues $\lambda_0$ such that $f'(\lambda_0)=0$, (ii) eigenvalues $\lambda_i$ for which $f'(\lambda_i)\neq0$ and the index $i<R$ is fixed, and (iii) eigenvalues $\lambda_k$ for which $f'(\lambda_k)\neq0$ and the index $k\geq R$ is running. Analogously to the proof of Theorem \ref{thm:upper}, noting Remark \ref{rem:apparent shortcut} and its resolution, the above estimate follows from the bounds on edges and vertices that are summarized and proved in Lemma \ref{lem:final bounds} below, where we stress that all constants that arise are independent of the running indices.
\end{proof}

\begin{lem}\label{lem:final bounds}
Let $n\geq3$. Let $f,\{\lambda_k\}_{k=1}^\infty$ be as in Theorem \ref{thm:main2}. Let $\lambda_0\in\R$ be such that $f'(\lambda_0)=0$. There exist $c_1,c_2>0$ such that we have
\begin{multicols}{2}
\begin{enumerate}
\item $|f'\{\lambda_{k_1},\lambda_{k_2}\}|>c_1|\lambda_{k_1}|^{-p}$,
\item $|f'\{\lambda_{i},\lambda_{k}\}|>c_1$,
\item $|f'\{\lambda_{i_1},\lambda_{i_2}\}|>c_1$,
\item $|f'\{\lambda_0,\lambda_k\}|> c_1|\lambda_k|^{-p-1}$,
\item $|f'\{\lambda_0,\lambda_i\}|>c_1$,
\item $|f'\{\lambda_0,\lambda_0\}|>c_1$,
\item $|f'\{\lambda_{k_1},\ldots,\lambda_{k_n}\}|\leq c_2|\lambda_{k_1}|^{-p}$,
\item $|f'\{\lambda_0,\lambda_{k_2},\ldots,\lambda_{k_n}\}|\leq c_2|\lambda_{k_2}|^{-p-1}$,
\item $|f'\{\lambda_{j_1},\ldots,\lambda_{j_n}\}|\leq c_2$,
\item[\vspace{\fill}] 
\item[\vspace{\fill}] 
\item[\vspace{\fill}]
\end{enumerate}
\end{multicols}
\noindent
for all $i,i_1,i_2\leq i_R$, $k,k_1,\ldots,k_n\geq i_R$ such that $f'(\lambda_j)\neq0$ for $j\in\{i,i_1,i_2,k,k_1,\ldots,k_n\}$ and such that $|\lambda_{k_1}|\leq\cdots\leq|\lambda_{k_n}|$, and for all $j_1,\ldots,j_n\in\N$.
\end{lem}


\begin{proof}
The estimates 1.\ and 7.\ follow from Theorem~\ref{thm:weighted_divdiff_bound}, 4.\ and 8.\ follow from Lemma~\ref{lem:zero-index bounds}, and 9.\ was shown in Corollary~\ref{cor:wdivdifs are bounded}. Estimates 3., 5., and 6.\ follow from the fact that there is a finite amount of indices $i<i_R$, and for each $i_1,i_2<i_R$ the (weighted) divided difference is nonzero. {Estimate 2.\ follows from Theorem~\ref{thm:no_vanishing_divdifs}.}
%
\end{proof}
We conclude this section with some remarks and examples illustrating the influence of 0-faces on the behaviour of amplitudes. 

\begin{rem}
The following explains the absence of a lower bound in Theorem \ref{thm:main2}. Suppose $i_1,i_2,i_3$ are such that $f'[i_1,i_2,i_3]=0$, $f'(\lambda_{i_1})\neq0$, and $f'(\lambda_{i_3})\neq0$. We then compute the amplitude
\begin{align*}
\raisebox{-19.5pt}{
\begin{tikzpicture}[thick]
\draw (-0.5,0.5) to (0,0);
\draw (-0.5,-.5) to (0,0);
\draw (0,0) to[out=60,in=120] (1,0);
\draw (0,0) to[out=-60,in=-120] (1,0);
\draw (1,0) to (1.5,0);
\node at (0,0.5) {$i_1$};
\node at (-.5,0) {$i_2$};
\node at (0,-.5) {$i_3$};
\vertex{0,0};
\vertex{1,0};
\end{tikzpicture}
}
=-\lambda_{i_1}^{-1}\lambda_{i_2}^{-1}\lambda_{i_3}^{-1}\sum_{k=1}^N\frac{f'\{i_1,i_2,i_3,k\}f'\{i_1,i_3,k\}}{f'\{i_1,k\}f'\{i_3,k\}}.
\end{align*}
Contrary to the situation where $f'[i_1,i_2,i_3]\neq0$, the factor $|f'\{i_1,i_2,i_3,k\}|$ obtained from the 4-vertex is proportional to $|\lambda_k|^{-1}$ as $k\to\infty$.
The other factors are proportional to $1$ as usual. We obtain
\begin{align*}
\raisebox{-19.5pt}{
\begin{tikzpicture}[thick]
\draw (-0.5,0.5) to (0,0);
\draw (-0.5,-.5) to (0,0);
\draw (0,0) to[out=60,in=120] (1,0);
\draw (0,0) to[out=-60,in=-120] (1,0);
\draw (1,0) to (1.5,0);
\node at (0,0.5) {$i_1$};
\node at (-.5,0) {$i_2$};
\node at (0,-.5) {$i_3$};
\vertex{0,0};
\vertex{1,0};
\end{tikzpicture}}
\sim\sum_{k=1}^N|\lambda_k|^{-1}\sim\sum_{k=1}^Nk^{-\frac{1}{d}}<\O(N).
\end{align*}
The order depends on $d$ but (since $d\geq0$) at least it is smaller than $\O(N)$ in the sense that there exists no $c$ such that the lower bound is $\geq cN$. For $d<1$ the graph is in fact finite.
\end{rem}

\begin{rem}
Theorem \ref{thm:main2} shows that the order of divergence depends on whether the external modes are in~$(f')^{-1}\{0\}$. For example, the calculated divergence of
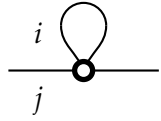
\begin{align*}
\raisebox{-20pt}{
\begin{tikzpicture}[thick]
\draw (-1,0) to (1,0);
\draw (-.3,.6) arc (180:0:.3cm);
\draw (0,0) to[out=110,in=-90] (-.3,.6);
\draw (0,0) to[out=70,in=-90] (.3,.6);
\vertex{0,0};
\node at (-.6,.5) {$i$};
\node at (-.6,-.4) {$j$};
\end{tikzpicture}}
\end{align*}
is $N^{1}$ if $f'(\lambda_i)\neq0$ and $N^{1+(p+1)/d}$ if $f'(\lambda_i)=0$. The latter is only an upper bound, but is actually attained for a function $f$ that is defined outside a compact as $f(x)=1/|x|^{p}$ and satisfies $(f')^{-1}\{0\}=\{0\}$. 
\end{rem}

\begin{rem}
Even though there exist only finitely many eigenmodes $\lambda_k$ with $f'(\lambda_k)=0$, these singular modes can boost the order of divergence not only when occurring as external indices. For instance, assuming non-singular external indices, we have the divergences
\begin{align*}
\raisebox{-20pt}{
\begin{tikzpicture}[thick]
\draw (-1,0) to (1,0);
\draw (-.7,.7) arc (180:0:.7cm);
\draw (0,0) to[out=160,in=-90] (-.7,.7);
\draw (0,0) to[out=20,in=-90] (.7,.7);
\draw (-.3,.6) arc (180:0:.3cm);
\draw (0,0) to[out=110,in=-90] (-.3,.6);
\draw (0,0) to[out=70,in=-90] (.3,.6);
\vertex{0,0};
\end{tikzpicture}}
=\O(N^{2+\frac{p}{d}})
\quad\text{and}\quad
\raisebox{-20pt}{
\begin{tikzpicture}[thick]
\draw (-1,0) to (1,0);
\draw (-.7,.7) arc (180:0:.7cm);
\draw (0,0) to[out=160,in=-90] (-.7,.7);
\draw (0,0) to[out=20,in=-90] (.7,.7);
\draw (-.3,.6) arc (180:0:.3cm);
\draw (0,0) to[out=110,in=-90] (-.3,.6);
\draw (0,0) to[out=70,in=-90] (.3,.6);
\vertex{0,0};
\node at (0,.5) {$0$};
\end{tikzpicture}}
=\O(N^{1+\frac{p+1}{d}}).
\end{align*}
The latter is larger than the former precisely if $d<1$. More generally, at second loop order, the boost of UV-divergence by internal singular indices is not apparent for $d\geq1$. Indeed, if $d\geq1$, then the diagrams of Remark \ref{rem:2-point 2-loop} remain precisely those of maximal order, also in the more general setting of Theorem \ref{thm:main2}. 
However, for any $d\in\N$, at loop order $L=d+2$ the maximal diagrams become those where one of the faces is artificially broken by a singular index, such as
\begin{align}
\raisebox{-35pt}{
\begin{tikzpicture}[thick]
\draw (-1.5,0) to (-1,0);
\draw (-1,0) arc (-180:180:1cm);
\draw (-1,0) to (0.5,.866);
\draw (0.5,.866) to (0.5,-.866);
\draw (-1,0) to (0.5,-.866);
\draw (0.5,.866) to (.75,1.3);
\draw (0.5,-.866) to (.75,-1.3);
\vertex{-1,0};
\vertex{0.5,.866};
\vertex{0.5,-.866};
\end{tikzpicture}}
~~
=\O(N^{4+3p/d})\quad\text{and}\quad
\raisebox{-35pt}{
\begin{tikzpicture}[thick]
\draw (-1.5,0) to (-1,0);
\draw (-1,0) arc (-180:180:1cm);
\draw (-1,0) to (0.5,.866);
\draw (0.5,.866) to (0.5,-.866);
\draw (-1,0) to (0.5,-.866);
\draw (0.5,.866) to (.75,1.3);
\draw (0.5,-.866) to (.75,-1.3);
\node at (0,0) {$0$};
\vertex{-1,0};
\vertex{0.5,.866};
\vertex{0.5,-.866};
\end{tikzpicture}}
=\O(N^{3+3(p+1)/d}).\label{eq:3-point 4-loop}
\end{align}
The latter is larger than the former precisely if $d<3$. Graphs similar to the examples above yield divergences increased by breaking faces to singular modes for any $d<L-1$.
\end{rem}
\begin{rem}\label{rem:conjectures}
In the setting of Theorem \ref{thm:main2} and for $d<3$, the second graph of \eqref{eq:3-point 4-loop} has maximal order among the 3-point 4-loop diagrams, but not every 3-point 4-loop diagram that has maximal order in the setting of Theorem \ref{thm:main} has maximal order in the setting of Theorem \ref{thm:main2}. That being said, we conjecture that, among the $n$-points $L$-loop graphs, every graph with maximal order in the setting of Theorem \ref{thm:main2} has maximal order in the setting of Theorem \ref{thm:main}. We moreover conjecture that the Ward identity, in the sense of \cite{vNvS22b}, holds in both cases when restricting to the graphs of maximal order of divergence.
This might help generalize the results of \cite{vNvS22b} to higher loop, which is a pressing open problem.
\end{rem}

\begin{rem}\label{rem:UV/IR}
Smooth even functions $f$ satisfy $f'(0)=0$, and if $\{\lambda_k\}$ is the spectrum of a typical Dirac operator, these eigenvalues correspond to modes of zero momentum. The fact that UV-divergences can become higher as external momenta vanish is reminiscent of noncommutative quantum field theory, cf.\ \cite{MRS2000}. For instance, in the naive version of noncommutative $\phi^4_4$, similar behavior was shown to lead to UV/IR mixing and, consequently, nonrenormalizability \cite{CR2000,CR2001,GW2005a}. This prompted the Grosse--Wulkenhaar model which solved the UV/IR problem \cite{GW2005b} and proved incredibly successful \cite{DGMR2007,GW2014}.



In light of this, it seems instructive to determine whether UV/IR mixing is present in the spectral action matrix model and its relatives. This is a difficult question to answer,
requiring a careful renormalization analysis and the passage to continuous spectrum.
\end{rem}

\begin{rem}
If $D$ has continuous spectrum, one may hope to define amplitudes in a (UV- and IR-) regulator independent way, probably by using multiple operator integrals. Roughly, sums like $\sum_{k=1}^N |\lambda_k|^r\asymp N^{(r/d)+1}$ should be exchanged for integrals like $\int_{[-M,M]^d} \mathrm{d}^d k\|\slashed{k}\|^r\asymp M^{r+d}$ for $M=N^{1/d}$. We expect similar power counting formulas to hold.
\end{rem}

\begin{rem}
It is surprising at first that faster decay of $f$ (larger $p$) corresponds to higher divergence of the graphs. In usual models the free theory scales quadratically with momentum (e.g., the Maxwell action of QED has $c(p)=-\eta^{\mu\nu}p^2+p^\mu p^\nu$ at momentum mode $p$) which by using Gaussian integration $\frac{\int {\rm d} x x^2e^{-\frac{1}{2}c(p)x^2}}{\int {\rm d} x e^{-\frac{1}{2}c(p)x^2}}=c(p)^{-1}$ leads to a propagator of order $\O(p^{-2})$. The second derivative of the spectral action contributes $f'[k,l]$ which for more rapidly decaying $f$ also decays more rapidly for large $k$ and $l$, and by Gaussian integration leads to a propagator which diverges for large modes. Because of the structure of divided differences this divergence is not fully compensated by the vertices.
\end{rem}

\begin{rem}
If we do not restrict the graphs to a fixed loop order, then we should note that the power counting degree $\omega(G)$ is unbounded over all graphs $G$. Indeed, to any graph with an unbroken face we may add arbitrary many self-loops which increase $U$ and $\Efi$ without altering $\Vfi$. Restricting to a certain loop order is done in perturbative QFT by using formal power series (throughout renormalization \textit{et cetera}). For the actions $S_1[V]=\hbar^{-1}\Tr(f(D+V)-f(D))\in\C[[\hbar]]$ and $S_2[V]=\hbar^{-2}\Tr(f(D+\hbar V)-f(D))\in\C[[\hbar]]$ the order in $\hbar$ of a certain graph $G$ corresponds 1-to-1 with $L$. The latter action fits well with the path integral interpretation and does not require the noncommutative Taylor series to converge. However, the heavy dependence on a formal parameter (of the boundedness of the power counting degree, for instance) indicates a difficulty for the passage to constructive QFT, which is a secondary motivation for this work. New insights are needed for this formidable challenge.
\end{rem}

\begin{rem}
It would be great to have a power counting formula when $f$ is a polynomial, like in \cite{BCEG}. Although similar techniques may be employed, we expect the result to be quite different.
\end{rem}

\section*{Conflict of Interest}
The authors declare that there is no conflict of interest.

\section*{Data Availability}
No data were used or generated in this study.

\bibliographystyle{abbrv}
\bibliography{bibliography}
\end{document}